\documentclass[11pt]{article}

\usepackage{amsmath,amsthm}
\usepackage{algorithm}
\usepackage{algpseudocode}
\usepackage{graphicx} 
\usepackage{nicefrac}
\usepackage{braket}
\newcommand{\id}{\mathsf{id}}

\newcommand{\CONGEST}{\textsf{CONGEST}}
\newcommand{\E}{\mathbb{E}}
\DeclareMathOperator{\poly}{poly}
\DeclareMathOperator{\polylog}{polylog}
\newcommand{\F}{{\mathsf{OUT}}}
\renewcommand{\H}{{\mathsf{IN}}}

\newcommand{\eps}{\varepsilon}
\newcommand{\leader}{v_{\mathsf{lead}}}

\newcommand{\setup}{\mathsf{Setup}}
\newcommand{\checking}{\mathsf{Checking}}
\newcommand{\find}{\mathsf{found}}

\newcommand{\qGamma}{q}

\newtheorem{theorem}{Theorem}
\newtheorem{lemma}{Lemma}

\newtheorem{claim}{Claim}
\newtheorem{fact}{Fact}
\newenvironment{claimproof}[1]{\par\noindent\emph{Proof of claim.}\space#1}{\hfill$\diamond$\medskip}

\newboolean{all}
\newboolean{short}
\setboolean{all}{false}
\setboolean{short}{false}
\ifthenelse{\boolean{all}}{
\usepackage[bibliography=common]{apxproof}
\newcommand{\withappendix}[1]{{\color{red} #1}}
\newcommand{\withoutappendix}[1]{{\color{cyan} #1}}
}{
\ifthenelse{\boolean{short}}{
\usepackage[bibliography=common]{apxproof}
\newcommand{\withappendix}[1]{#1}
\newcommand{\withoutappendix}[1]{}
}
{
\usepackage[appendix=inline,bibliography=common]{apxproof}
\newcommand{\withappendix}[1]{}
\newcommand{\withoutappendix}[1]{#1}
}
}

\usepackage{fullpage,authblk,enumitem,amssymb,xcolor}
\setlist{itemsep=0em,topsep=.5em,parsep=.2em} 

\begin{document}

\title{Even-Cycle Detection in the Randomized and Quantum CONGEST Model%
\thanks{Research supported in part by the 
European QuantERA project QOPT (ERA-NET Cofund 2022-25),
the French ANR projects DUCAT (ANR-20-CE48-0006) and QUDATA (ANR-18-CE47-0010), and
the French PEPR integrated project EPiQ (ANR-22-PETQ-0007). This work has received support under the program "Investissement d'Avenir" launched by the French Government and implemented by ANR, with the "ANR‐21‐CMAQ-0001, FQPS".}}
\author[1]{Pierre Fraigniaud}
\author[1]{Maël Luce}
\author[1]{Frédéric Magniez}
\author[2]{Ioan Todinca}
\affil[1]{
  {Université Paris Cité, CNRS, IRIF},
  {Paris},
  {France}
}
\affil[2]{
  {Université d'Orléans, INSA-Centre Val de Loire, LIFO},
 {Orléans},
  {France}   
}
\date{}

\maketitle 

\begin{abstract}
We show that, for every $k\geq 2$, $C_{2k}$-freeness can be decided in $O(n^{1-1/k})$ rounds in the \CONGEST{} model by a randomized Monte-Carlo distributed algorithm with one-sided error probability $\nicefrac13$. This matches the best round-complexities of previously known algorithms for $k\in\{2,3,4,5\}$ by Drucker et al. [PODC'14] and Censor-Hillel et al. [DISC'20], but improves the complexities of the known algorithms for $k>5$ by  Eden et al. [DISC'19], which were essentially of the form $\tilde O(n^{1-2/k^2})$. Our algorithm uses colored BFS-explorations with threshold, but with an original \emph{global} approach that enables to overcome a recent impossibility result by Fraigniaud et al. [SIROCCO'23] about using colored BFS-exploration with \emph{local} threshold for detecting cycles.

We also show how to quantize our algorithm for achieving a round-complexity $\tilde O(n^{\nicefrac{1}{2}-\nicefrac{1}{2k}})$ in the quantum setting for deciding $C_{2k}$ freeness. Furthermore, this allows us to  improve the known quantum complexities of the simpler problem of detecting cycles of length \emph{at most}~$2k$ by van Apeldoorn and de Vos [PODC'22]. Our quantization is in two steps. First, the congestion of our randomized algorithm is reduced, to the cost of reducing its success probability too. Second, the success probability is boosted using a new quantum framework derived from sequential algorithms, namely Monte-Carlo quantum amplification.
\end{abstract}

\section{Introduction}\label{sec2}

For every fixed graph $H$, $H$-freeness is the problem consisting in deciding whether any given graph~$G$ contains $H$ as a subgraph. In the distributed setting, in which each vertex of~$G$ is a computing element, the decision rule is specified as: $G$ is $H$-free if and only if all nodes of $G$ accept. That is, if $H$ is a subgraph of~$G$, then at least one node of~$G$ must reject, otherwise all nodes must accept. 

$H$-freeness has been extensively studied in various distributed models, including the standard \CONGEST\/ model (see the recent survey~\cite{Censor-Hillel21,Censor-Hillel22}). Recall that, in this latter model, a network is modeled as a simple connected $n$-vertex graph~$G$, where the computing nodes occupy the vertices of~$G$, and they exchange messages along the edges of~$G$. All nodes start computing at the same time, and they perform in lockstep, as a sequence of synchronous rounds. At each round, every node sends a (possibly different) message to each of its neighbors in~$G$, receives the message sent by its neighbors, and performs some local computation. The main constraint imposed on computation by the \CONGEST\/ model is on the size of the messages exchanged during each round, which must not exceed $O(\log n)$ bits in $n$-node networks~\cite{Peleg2000}.

The specific case of \emph{cycle-detection}, i.e., deciding $C_k$-freeness for a fixed $k\geq 3$, has  been the focus of several contributions to the \CONGEST\/ model (cf. Table~\ref{tab:summary}) --- unless specified otherwise, all the algorithms mentioned there are in the randomized setting, i.e., they are Monte-Carlo algorithms with 1-sided error probability~$\nicefrac13$.
More generally, for $\eps\in(0,1)$, a randomized distributed algorithm~$\mathcal{A}$ solves $C_{k}$-freeness with one-sided error probability $\eps$ if, for every input graph~$G$,
\begin{itemize}
    \item If $G$ contains a cycle $C_{k}$ as a subgraph, then, with probability at least $1-\eps$, at least one node of $G$ outputs reject; 
    \item Otherwise, all nodes of $G$ output accept with probability~$1$.
\end{itemize}
For odd cycles of length at least~5, the problem is essentially solved. Indeed, for every $k\geq 3$, $C_k$-freeness can be decided in $O(n)$ rounds by a \emph{deterministic} algorithm~\cite{KorhonenR17}, and it was shown that, for every odd $k\geq 5$, deciding $C_k$-freeness requires $\tilde\Omega(n)$ rounds, even for randomized algorithms~\cite{DruckerKO13}. The case of triangles, i.e., deciding $C_3$-freeness, is however still open. The best known upper bound on the round-complexity of deciding $C_3$-freeness is $\tilde O(n^{1/3})$~\cite{chang2019improved}. However, it was shown that any polynomial lower bound for this problem, i.e., any bound of the form $\Omega(n^\alpha)$ with $\alpha>0$ constant, would imply major breakthroughs in circuit complexity~\cite{EdenFFKO19}.

\begin{table}[ht]
\vspace*{-0.5em}\centering
    \begin{tabular}{c|c|c|c|}
    Reference & Cycle & Complexity & Framework\\
    \hline 
    \cite{chang2019improved} & $C_3$ & $\tilde O(n^{1/3})$ & rand.    \\
    \cite{DruckerKO13,KorhonenR17} & $C_{2k+1}, k\geq 2$ & $\tilde\Theta(n)$  & $O$ det. / $\tilde\Omega$ rand. \\
    \hline 
    \cite{DruckerKO13} & $C_4$ & $\tilde\Theta(\sqrt{n})$  & $O$ det. / $\tilde\Omega$ rand. \\
    \cite{KorhonenR17} & $C_{2k}, k\geq 2$ & $\tilde\Omega(\sqrt{n})$  & rand.\\
    \cite{Censor-HillelFG20} & $C_{2k}, k\in\{2,3,4,5\}$ & $O(n^{1-1/k})$ & rand. \\
    \cite{EdenFFKO19} & $C_{2k}$, $k\geq 6$ even &  $\tilde O(n^{1-2/(k^2-2k+4)})$ & rand. \\
    \cite{EdenFFKO19} & $C_{2k}$, $k\geq 7$ odd & $\tilde O(n^{1-2/(k^2-k+2)})$  & rand. \\
     \cite{Censor-HillelFG20} & $\{C_\ell \mid 3 \leq \ell \leq 2k\}, k\geq 2$ & $\tilde O(n^{1-1/k})$ & rand.\\
     \hline 
    this paper & $C_{2k}, k\geq 2$ & $O(n^{1-1/k})$  &  rand.\\
     \hline 
    \cite{Censor-HillelFG22} & $C_3$ & $\tilde{O}(n^{1/5})$ & quant.\\
    \cite{censor2023discussion} & $C_4$ & $\tilde O(n^{1/4})$ & quant.\\
    \cite{ApeldoornV22} & $\{C_\ell \mid 3 \leq \ell \leq 2k\}, k\geq 2$ & $\tilde O(n^{\nicefrac12-\nicefrac{1}{4k +2}})$ & quant.\\
\hline 
   this paper & $C_{2k}, k\geq 2$ & $\tilde O(n^{\nicefrac{1}{2}-\nicefrac{1}{2k}})$  & quant. \\
   this paper & $C_{2k}, k\geq 2$ & $\tilde\Omega(n^{1/4})$  &  quant.\\
    this paper & $C_{2k+1}, k\geq 2$ & $\tilde\Theta(\sqrt{n})$  & quant.\\
   this paper & $\{C_\ell \mid 3 \leq \ell \leq 2k\}, k\geq 2$ & $\tilde O(n^{\nicefrac{1}{2}-\nicefrac{1}{2k}})$ & quant.\\
 \hline 
    \end{tabular}
    \vspace*{-0.6em}
    \caption{\vspace*{-2em}Summary of the results about deciding $C_k$-freeness in \CONGEST.}
    \label{tab:summary}
\end{table}

The landscape of deciding the presence of an even-size cycle is more contrasted. It was first shown that $C_4$-freeness can be decided in $O(\sqrt{n})$ rounds, and that this complexity is essentially optimal thanks to a lower bound of $\tilde\Omega(\sqrt{n})$ rounds~\cite{DruckerKO13}. It was then shown that the lower bound of $\tilde\Omega(\sqrt{n})$ rounds applies to deciding $C_{2k}$-freeness too, for every $k\geq 2$~\cite{KorhonenR17}. However, any lower bound of the form $\Omega(n^{1/2+\alpha})$ with $\alpha>0$ for deciding $C_{2k}$-freeness  for some $k\geq 3$ would imply new lower bounds in circuit complexity, which are considered hard to obtain (see~\cite{Censor-HillelFG20}). On the positive side, it was first proved that, for every $k\geq 2$,  $C_{2k}$-freeness can be decided in $O(n^{1-1/k(k-1)})$ rounds~\cite{FischerGKO18}. This was later improved in~\cite{EdenFFKO19}, where it is proved that $C_{2k}$-freeness can be decided in $O(n^{1-2/(k^2-2k+4)})$ for even $k\geq 4$, and in $O(n^{1-2/(k^2-k+2)})$ for odd $k\geq 3$. 
However, better results were obtained for small cycles, with algorithms for deciding $C_{2k}$-freeness performing in $O(n^{1-1/k})$ rounds, for $k\in\{3,4,5\}$~\cite{Censor-HillelFG20}. These algorithms were obtained by analyzing the congestion caused by colored BFS-explorations, and by showing that if the congestion exceeds a certain threshold then there must exist a $2k$-cycle in the network. Unfortunately, it was later shown (see~\cite{FraigniaudLT23}) that this technique does not extend to larger cycles, of length $2k$ for $k\geq 6$. 

A problem related to deciding $C_k$-freeness is deciding whether there is a cycle of length \emph{at most}~$2k$. Surprisingly, this problem is actually simpler, because deciding $C_k$-freeness can benefit from the fact that there are no cycles of length less than~$k$. This fact was exploited in \cite{Censor-HillelFG20} to design an algorithm deciding $\{C_\ell \mid 3 \leq \ell \leq 2k\}$-freeness in $O(n^{1-1/k})$ rounds, for any $k\geq 2$. 

Moreover, the \CONGEST\/ model can take benefits of quantum effects, as far as detecting subgraphs is concerned. Indeed, it was proved that deciding $C_3$-freeness (respectively, $C_4$-freeness) can be decided in $\tilde{O}(n^{1/5})$ rounds~\cite{Censor-HillelFG22} (resp., $\tilde{O}(n^{1/4})$ rounds~\cite{censor2023discussion}) in the quantum \CONGEST\/ model. Similarly,  $\{C_\ell \mid 3 \leq \ell \leq 2k\}$-freeness can be decided in $\tilde O(n^{\nicefrac{1}{2}-\nicefrac{1}{4k+2}})$ rounds~\cite{ApeldoornV22} in the quantum \CONGEST\/ model. However, since any lower bound on the round-complexity of a problem in quantum \CONGEST\/ is also a lower bound on the round-complexity of the same problem in the classical version of \CONGEST, the hardness of designing lower bounds for triangle-freeness and for $C_{2k}$-freeness for $k\geq 3$ in the classical setting also holds in the quantum setting. 

\subsection{Our Results}
\label{subsec:OurResults}

Our results are summarized in Table~\ref{tab:summary}. 
We considerably simplify the landscape of results about deciding $C_{k}$-freeness, and we extend it to the quantum \CONGEST\/ framework. 

Our first main contribution shows that, for every $k\geq 2$, the complexity of deciding $C_{2k}$-freeness is $O(n^{1-1/k})$, hence extending the results of Drucker et al.~\cite{DruckerKO13} and Censor-Hillel et al.~\cite{Censor-HillelFG20} to all $k\geq 6$, and improving the complexity of the algorithms by Eden et al.~\cite{EdenFFKO19}. 
Specifically, we show the following. 

\begin{theorem}\label{classical_theo}
For every  integer $k\geq 2$,  and every real $\eps>0$, there is an algorithm that solves $C_{2k}$-freeness with one-sided error probability $\eps$ in $O( \log^2(1/\eps) \cdot 2^{3k}\, k^{2k+3} \cdot n^{1-1/k})$ rounds in the \CONGEST\/ model.
\end{theorem}

Our second main contribution is related to analysing the speedup that can be obtained by allowing nodes to handle entangled quantum resources, yet exchanging at most $O(\log n)$ qubits at each round. We show that quantum resources enable to obtain a quadratic speedup for deciding cycle-freeness. Specifically, we show the following.

\begin{theorem}\label{quantum_theo_simplified}
In the quantum \CONGEST\/ model, the following holds:
\begin{itemize}
\item For every  integer $k\geq 2$, there is a quantum algorithm deciding $C_{2k}$-freeness with one-sided error probability $1/\mbox{\rm\small poly}(n)$ in $k^{O(k)}\cdot \mbox{\rm polylog}(n)\cdot n^{\nicefrac12-\nicefrac{1}{2k}}$ rounds, and any algorithm deciding $C_{2k}$-freeness with one-sided error probability at most $1/3$ performs in at least $\tilde\Omega(n^{1/4})$ rounds. 

\item For every  integer $k\geq 1$, there is a quantum algorithm deciding $C_{2k+1}$-freeness with one-sided error probability $1/\mbox{\rm\small poly}(n)$ in $\tilde O(\sqrt{n})$ rounds, and, for every  integer $k\geq 2$, any quantum algorithm deciding $C_{2k+1}$-freeness with one-sided error probability at most  $1/3$ performs in at least $\tilde\Omega(\sqrt{n})$ rounds. 
\end{itemize}
\end{theorem}

Therefore, our quantum algorithm for deciding $C_4$-freeness performing in $\tilde O(n^{1/4})$ rounds is, up to logarithmic multiplicative factors, optimal. For $k\geq 3$, we face the same difficulty as in the classical framework for designing lower bounds of deciding $C_{2k}$-freeness.  For odd cycles, we show that quantum algorithms enable quadratic speedup too, and this is essentially optimal, up to logarithmic factors. That is, the complexity of deciding $C_{2k+1}$-freeness in the quantum \CONGEST\/ model is $\tilde\Theta(\sqrt{n})$ rounds for all $k\geq 2$. As for the (classical) \CONGEST\/ model, the case of $C_3$-freeness remains open in absence of any non trivial lower bound. However, we point out that our new quantum framework enables to improve the quantum complexity of detecting cycles of length \emph{at most} $2k$ from~\cite{ApeldoornV22}. For every $k\geq 2$, our algorithm for $\{C_\ell \mid 3 \leq \ell \leq 2k\}$-freeness performs in $\tilde O(n^{\nicefrac{1}{2}-\nicefrac{1}{2k}})$ instead of $\tilde O(n^{\nicefrac{1}{2}-\nicefrac{1}{4k +2}})$. 

\subsubsection{Our Technique for \CONGEST}

As in previous papers about deciding $C_{2k}$-freeness~\cite{Censor-HillelFG20,eden2022sublinear}, we separate the case of detecting \emph{light} cycles from the case of detecting \emph{heavy} cycles, where the former are $2k$-cycles containing solely nodes the degrees of which do not exceed a specific bound~$d_{max}$, and the latter are $2k$-cycles that are not light. We use the same technique as~\cite{Censor-HillelFG20} for detecting light cycles, with the same bound $d_{max}=n^{1/k}$ on the degrees. Our main contribution is a new technique for detecting heavy cycles, i.e., cycles containing at least one node of degree larger than~$n^{1/k}$. 

A central technique for detecting $2k$-cycles is \emph{color coding}~\cite{AlonYZ16}, which is implemented in the distributed setting as the so-called \emph{colored BFS-exploration}~\cite{EvenFFGLMMOORT17,FraigniaudMORT17}. Roughly, it consists in each node~$v$ picking a color $c(v)\in \{0,\dots,2k-1\}$ uniformly at random (u.a.r.), and the goal is to check whether there is a cycle $C=(u_0,\dots,u_{2k-1})$ with $c(u_i)=i$ for every $i\in \{0,\dots,2k-1\}$. For this purpose, some nodes colored~0 launch a BFS-exploration, sending their identifiers to their neighbors colored~1 and $2k-1$. Then, at every round $i\in \{1,\dots,k-1\}$, nodes colored $i$ forward all the identifiers received from neighbors colored $i-1$ to their neighbors colored $i+1$, while nodes colored $2k-i$ forward all the identifiers received from neighbors colored $2k-(i-1)$ to their neighbors colored $2k-(i+1)$. If a node colored~$k$ receives the same identifier from a neighbor colored~$k-1$ and $k+1$, a $2k$-cycle has been detected. By repeating about $(2k)^{2k}$ times the random choice of the colors, the probability that an existing $2k$-cycle is well colored is at least~$2/3$, which guarantees the desired success probability. Using this approach, the issue is to control congestion, that is, to control the number of distinct identifiers that a same node has to forward during the BFS-exploration. 

For controlling the congestion in colored BFS-exploration, an elegant approach has been presented in~\cite{Censor-HillelFG20}, that we refer to as \emph{local threshold}. In essence, it consists in selecting a single source node~$s$ u.a.r., which triggers all its neighbors colored~0, asking them to launch a colored BFS-exploration. A key point is that, as shown in~\cite{Censor-HillelFG20}, for $k\in\{2,3,4,5\}$, there exists a (constant) threshold $\tau_k\geq 1$ such that, a constant fraction of sources~$s$ are either in a $2k$-cycle or will not cause any node to receive more than $\tau_k$ identifiers. Therefore, each attempt to find a $2k$-cycle triggered by a selected source~$s$ takes a constant number of rounds (namely, at most $k\cdot \tau_k$ rounds). Now, the probability that the selected source $s$ has a neighbor in a cycle is at least $1/n^{1/k}$, and thus it is sufficient to repeat $O(n^{1-1/k})$ times the random choice of~$s$ for finding a $2k$-cycle with constant probability (if it exists). Unfortunately, the local threshold technique was proved to suffer from some limitation~\cite{FraigniaudLT23}, namely it is not  extendable to the detection of larger $2k$-cycles, i.e., to $k>5$. 

To overcome the inherent limitation of the local threshold technique, we adopt a \emph{global threshold} approach. Specifically, instead of repeating $O(n^{1-1/k})$ times the choice of a single source~$s$ that triggers the colored BFS-exploration, we directly select a set $S$ of $O(n^{1-1/k})$ random sources, and we only need to repeat $O(1)$ times this choice. Now,  we show that there exists a (global) threshold $\tau_k(n)=O(n^{1-1/k})$ such that, for each choice of~$S$, if a node has to forward more than $\tau_k(n)$ identifiers at some round, then there must exist a $2k$-cycle. Establishing this result is the key point in our algorithm. It requires to analyze in detail the successive rounds of the BFS-exploration which may eventually cause a node~$v$ to receive a set $I_v$ of more that $\tau_k(n)$ identifiers, and to use this analysis for constructing a $2k$-cycle. This cycle is itself the union of a path $P$ alternating between nodes in~$S$ and nodes whose identifiers are in~$I_v$, and two vertex-disjoint paths connecting the two end points of~$P$ to~$v$. 

\subsubsection{Our Technique for Quantum \CONGEST} 

One of the main tools to speed-up distributed algorithm in the quantum framework is the \emph{distributed Grover search}~\cite{GallM18}. Indeed, (sequential) Grover search enables solving problems such as minimizing a function, or searching for a preimage, in a time-complexity that is often quadratically faster than the classical (i.e., non quantum) time-complexity. Transferred to the quantum \CONGEST\/ model,  Grover search was first used to design  sublinear algorithms for computing the diameter of a graph~\cite{GallM18}.
A nested version of Grover search was also used recently for detecting cliques~\cite{Censor-HillelFG22} and a parallel one for cycles~\cite{ApeldoornV22}.

We take one step further in the use of Grover search, by defining an encapsulated quantum framework for distributed computing that we call \emph{distributed quantum Monte-Carlo amplification} (Theorem~\ref{theo:distampli}). We show that, given any distributed (quantum, or randomized) Monte-Carlo algorithm with small one-sided \emph{success} probability~$\eps$ and round-complexity~$R$,  there exists a quantum algorithm with constant (e.g., $\nicefrac23$)  one-sided \emph{success} probability, whose round-complexity is roughly~$\sqrt{\nicefrac{1}{\eps}}\cdot R$. Observe that boosting the success probability to a constant would have required
${1}/{\eps}$ iterations in a non-quantum setting. Our boosting technique is quadratically faster.
There is however a cost to pay. Indeed, this is ignoring an additional term 
${D}/\sqrt{\eps}$, where $D$ is the diameter of the graph, that appears in the overall quantum complexity. This may not be an issue for problems whose round-complexity inherently depends on the diameter (e.g., computing an MST), but this is an issue for ``local problems'' such as $H$-freeness. Nevertheless, we can eliminate the diameter dependence by employing the standard \emph{diameter reduction} technique~\cite{eden2022sublinear}.

In order to apply our amplification technique, we first  \emph{decrease} the success probability of our classical algorithm. This allows us to decrease its congestion too.
Specifically, our classical algorithm has constant success probability, and congestion $O(n^{1-1/k})$. 
Indeed, in colored BFS-exploration performed by the classical algorithm, every node participating to the exploration forwards information that comes from at most $\tau_k(n)$ nodes, where $\tau_k(n)=O(n^{1-1/k})$ is the global threshold. This yields a congestion at most~$\tau_k(n)$, and thus a round complexity at most~$\tau_k(n)$. By not activating systematically the nodes of the tested set in the colored BFS-exploration, but by activating each of them independently with probability~$\eps=O(\nicefrac{1}{\tau_k(n)})$, we show that the congestion decreases to $O(\eps\,\tau_k(n))=O(1)$, and the success probability drops down to~$\eps$. At this point, we can apply our Monte-Carlo amplification technique to get one sided-error probability $\nicefrac{1}{\mathrm{poly}(n)}$ with round-complexity $\tilde O(D\cdot \sqrt{\tau_k(n)})$ rounds, 
where $D$ is the  diameter of the graph. Finally, the diameter dependence is eventually removed, up to $\mathrm{polylog}$ factors, by using the technique in~\cite{eden2022sublinear}.

\subsection{Related Work}
\label{subsec:RelatedWork}

There is a vast literature on distributed algorithms for deciding whether the input graph~$G$ includes a  fixed given graph~$H$ as a subgraph, or as an induced subgraph. We refer to the recent survey \cite{Censor-Hillel21,Censor-Hillel22} for a detailed description of the recent progress in the matter, including the typical techniques. We just clarify here a few points that sometimes cause confusion. First, deciding $H$-freeness is also referred to as \emph{detecting}~$H$. Second, subgraph detection has a more demanding variant, called subgraph \emph{listing}. In the latter, each occurrence of~$H$ must be reported by at least one node. Subgraph detection and subgraph listing each have their \emph{local} variants. Local detection imposes that each node outputs accept or reject based on whether it is a part of a copy of $H$ or not, and local listing imposes that each node outputs the list of occurrences of the subgraph~$H$ it belongs to. In this paper, we focused on deciding $C_k$-freeness, i.e., cycle detection. Besides cycles, two families of graphs have been extensively studied in the framework of distributed $H$-freeness, trees~\cite{KorhonenR17} and cliques~\cite{Censor-HillelCG21,CzumajK18,EdenFFKO19}. Subgraph detection as well as subgraph listing have been studied in other distributed computing models, such as \textsf{CONGESTED CLIQUE} (see, e.g., \cite{DolevLP12}). The \emph{testing} variant of the problem has also been considered, in which the goal is to decide whether the input graph is $H$-free or contains \emph{many} copies of~$H$ (measured, e.g., by the number of edges of $G$ that must be deleted to obtain a graph that is $H$-free). In this latter framework, it is possible to design decision algorithms performing in~$O(1)$ rounds in \CONGEST\/ (see, e.g., \cite{EvenFFGLMMOORT17,FraigniaudO19}).

In the quantum setting, it was shown in~\cite{Elkin+PODC14}  that the quantum \CONGEST\/ model is not more powerful that the classical \CONGEST\/ model for many important graph-theoretical problems. Nonetheless, it was later shown in~\cite{GallM18} that computing the diameter of the network can be solved more efficiently in the quantum setting. Since then, other quantum speed-ups have been discovered, including subgraph detection~\cite{Censor-HillelFG22,ApeldoornV22}. We also mention that a speed-up similar than ours has already been established for detecting~$C_4$, with round complexity $\tilde O(n^{1/4})$ in an unpublished work~\cite{censor2023discussion}. This latter contribution is directly using a distributed Grover search~\cite{GallM18}, and not our Monte-Carlo amplification.  The implementation of Grover in~\cite{censor2023discussion} is thus a bit more involved. In particular, a leader node is used for deciding which nodes are activated or not, based on a token sampling in $[1,\sqrt{n}]$. While the token sampling is decentralized, the activation of the nodes for a given token is centralized before applying  the distributed Grover search, leading to a multiplicative term equal to the diameter~$D$. Finally, we note that in the \textsf{CONGESTED CLIQUE} model as well, quantum algorithms faster than the best known classical algorithms have been discovered, starting for the All-Pair Shortest Path problem~\cite{Izumi+PODC19}, and recently for clique-detection~\cite{Censor-HillelFG22}.  In the \textsf{LOCAL} model, which is another fundamental model in distributed computing, separations between the computational powers of the classical and quantum versions of the model have been reported for some problems~\cite{Gavoille+DISC09,LeGall+STACS19}, but other papers have also reported that quantum effects do not bring significant benefits for other problems (e.g., approximate graph coloring~\cite{CoiteuxRoy2024}).

\section{Cycle Detection in the Classical CONGEST Model}\label{sec3}

This section is entirely dedicated to proving Theorem~\ref{classical_theo}. For this purpose, we first describe the algorithm claimed to exist in 
the statement of the theorem
(Algorithm~\ref{algo-C2k-freeness}), then we analyze it.

\subsection{Algorithm description}

 In Algorithm~\ref{algo-C2k-freeness}, $k\geq 2$ and $\epsilon>0$ are fixed parameters, and $G=(V,E)$ is the input graph. The only prior knowledge given to each node~$v\in V$, apart from $k$, $\epsilon$, and the identifier $\id(v)$, is the size~$n=|V|$ of the input graph. Algorithm~\ref{algo-C2k-freeness} is using a variant of \emph{color-BFS with threshold}~\cite{Censor-HillelFG20}, displayed as Procedure \textsf{color-BFS} in Algorithm~\ref{algo-C2k-freeness}, detailed next. 

\subsubsection{Procedure Color-BFS with Threshold} \label{subsec:color-BFS}

 The syntax of a call to color-BFS with threshold is 
$$
\textsf{color-BFS}(k,H,c,X,\tau),
$$
where $k$ is the fixed parameter, $H$ is a subgraph of the input graph~$G$, $c:V(H)\to \{0,\dots,2k-1\}$ is a (non-necessarily proper) coloring of the vertices of~$H$, $X\subseteq V(H)$ is a set of vertices, and $\tau\geq 0$ is an integer called \emph{threshold}. 
In all calls to $\textsf{color-BFS}$, graph $H$ will be an induced subgraph of $G$, and in particular every node of $G$ will locally know whether it belongs to $H$ or not. Similarly, every node will know whether it belongs to set $X$ or not.

In our variant of color-BFS, only the nodes $x\in X$ initiate the search for a $2k$-cycle, and the search is performed in~$H$ only, i.e., not necessarily in the whole graph~$G$. Specifically,  every node $x\in X$ colored~0, i.e., every node $x\in X$ for which $c(x)=0$, launches the search by sending its identifier $\id(x)$ to all its neighbors in~$H$ (cf. Instruction~\ref{ins:BFS-initialization}). 

Then, as specified in the for-loop at Instruction~\ref{ins:BFS-for-loop}, for every $i=1,\dots,k-1$, every node $v\in V(H)$ colored~$i$ having received a collection of identifiers $I_v\subseteq \{\id(x)\mid x\in X\}$ from its neighbors in $H$ colored~$i-1$ aims at forwarding $I_v$ to all its neighbors in $H$ colored~$i+1$. However, it does forward $I_v$ to its neighbors only if $I_v$ is not too large, namely only if the number  $|I_v|$ of identifiers in~$I_v$ does not exceed the threshold~$\tau$. Instead, if $v$ has collected too many identifiers, i.e., if $|I_v|>\tau$, then $v$ simply discards $I_v$, and forwards nothing. 

Similarly, for every $i=2k-1,\dots,k+1$, every node~$v\in V(H)$ colored~$i$ having received  a collection $I_v$ of identifiers from its neighbors in~$H$  colored~$i+1 \bmod 2k$ forwards $I_v$ to all its neighbors in $H$ colored~$i-1$ whenever $|I_v|\leq \tau$, and discards~$I_v$ otherwise. 

Eventually, at Instruction~\ref{ins:BFS-reject}, if a node $v\in V(H)$ colored~$k$ receives a same identifier $\id(x)$ of some node $x\in X$ from a neighbor in $H$ colored~$k-1$, and from a neighbor in $H$ colored~$k+1$, then $v$ rejects. Observe that when $v$ rejects, then by construction there is a $2k$-cycle containing $v$ and $x$. Indeed $\id(x)$ traveled to $v$ along two paths $(x=u_0,u_1,\dots,u_k=v)$ and $(x=u_0,u_{2k-1},\dots,u_k=v)$ of length $k$, with different internal vertices.

{
\begin{algorithm}[ht]
\caption{Deciding $C_{2k}$-freeness in $G=(V,E)$ 
with one sided-error $\eps$}
\label{algo-C2k-freeness}
\begin{algorithmic}[1]
\State $U\leftarrow \{u\in V \mid \deg(u)\leq n^{1/k}\}$; \Comment{$U$ is the set of ``light'' nodes}\label{ins:C2k-light-nodes}
\State $\hat\eps \gets \ln(3/\eps)$; \;  $p\gets \hat\eps \cdot 2k^2 / n^{1/k}$; \Comment{Setting of the selection probability $p$}
\State every node $u\in V$ selects itself with probability $p$; \label{inst:rdmS}
\State $S\gets \{\mbox{selected nodes}\}$; \Comment{In expectation, $|S|=np $}\label{inst:rdmS2}   
\State $W \leftarrow \{u\in V\smallsetminus S \mid |N_G(u)\cap S| \geq k^2\}$; \Comment{$u\in W$ iff $u$ has at least $k^2$ selected neighbors}\label{ins:Set-W}
\State $K\gets \hat\eps  \cdot (2k)^{2k}$; \; $\tau\gets k2^k \cdot np$; \Comment{Setting  of \#repetitions $K$, and threshold $\tau$}\label{ins:set-threshold} 
\For{$r=1$ \textbf{to} $K$}
\Comment{A sequence of $K$ search phases}\label{inst:nb_colorings}
    \State every node $u\in V$ picks a color $c(u)\in\{0,\dots,2k-1\}$ u.a.r.;\label{inst:coloring}
    \State $\textsf{color-BFS}(k,G[U],c,U,\tau)$; \label{inst:colorbfs-U} 
           \Comment{Search for $2k$-cycles  with light nodes only}
    \State $\textsf{color-BFS}(k,G,c,S,\tau)$; \label{inst:colorbfs-S}
          \Comment{Search for $2k$-cycles  with at least one selected node}
    \State $\textsf{color-BFS}(k,G[V\smallsetminus S],c,W,\tau)$; \label{inst:colorbfs-W}
          \Comment{Search for other $2k$-cycles}
\EndFor
\State every node that has not output \emph{reject} at a previous round outputs \emph{accept}.
\medskip
\Procedure{\rm\sf color-BFS}{$k,H,c,X,\tau$}
\State every node $x\in X$ with $c(x)=0$ sends its ID to its neighbors in~$H$ \label{ins:BFS-initialization}
\For{$i=1$ to $k-1$}\label{ins:BFS-for-loop}
\For{every node $v \in V(H)$ with $c(v)=i$ (resp., $c(v)=2k-i$) }
    \State $I_v\gets \{\mbox{IDs received from neighbors in~$H$ colored $i-1$ (resp. $2k-i-1)$}\}$ \label{ins:Iv}
    \If{$|I_v|\leq \tau$ }\label{ins:I_size}
        \State $v$ forwards $I_v$ to its neighbors in~$H$ colored~$i+1$ (resp. $2k-i-1)$
    \EndIf
\EndFor
\EndFor
\For{every node $v\in V(H)$ colored $k$} 
    \If {$v$ receives a same ID from two neighbors respectively colored~$k-1$ and~$k+1$}\label{ins:BFS-reject}
        \State $v$ outputs \emph{reject} 
    \EndIf
\EndFor
\EndProcedure
\end{algorithmic}
\end{algorithm}
}

\subsubsection{The Cycle Detection Algorithm}\label{subsec:algo}

Our cycle detection algorithm is detailed in Algorithm~\ref{algo-C2k-freeness}. It essentially consists of three calls to $\textsf{color-BFS}(k,H,c,X,\tau)$, for three different graphs~$H$, three different sets~$X$, and one threshold~$\tau$. The first set is the set $U$ of \emph{light} nodes in~$G$ (see Instruction~\ref{ins:C2k-light-nodes}), i.e., the set of nodes
\[
U=\{u\in V \mid \deg(u)\leq n^{1/k}\}.
\]
The second set is denoted by~$S$.  It is constructed randomly at Instructions~\ref{inst:rdmS}-\ref{inst:rdmS2}, by having each node independently deciding whether to enter set~$S$ with probability $p=\Theta(1/n^{1/k})$. In other words, if $B_p$ denotes the random variable with Bernoulli distribution of parameter~$p$, 
\[
S=\{u\in V \mid B_p(u)=1\}.
\]
Note that the expected size of $S$ is $\Theta(n^{1-1/k})$. As we shall show later, the probability~$p=\Theta(1/n^{1/k})$ can be set so that, for every node~$u$ of degree at least~$n^{1/k}$, the probability that~$u$ has at least a constant number $|N_G(u)\cap S|$ of neighbors in~$S$ is constant. This leads us to setting our third set, $W$, as the set of nodes that have at least $k^2$ selected neighbors (see Instruction~\ref{ins:Set-W}), that is, 
\[
W = \{u\in V\smallsetminus S \mid  |N_G(u)\cap S| \geq k^2\}.
\] 
Once these three sets $U$, $S$, and $W$ have been constructed, they remain fixed for the rest of the algorithm. Finally, a threshold 
\[
\tau=\Theta(n^{1-1/k})
\]
is set at Instruction~\ref{ins:set-threshold}. This threshold will be used in all color-BFS calls in the for-loop at Instruction~\ref{inst:nb_colorings}.  

The for-loop at Instruction~\ref{inst:nb_colorings} is performed a constant number $K$ of times. It will be shown later  that choosing $K$ as a sufficiently large constant will be sufficient to guarantee that, if there is a $2k$-cycle $C$ in~$G$, then, with constant probability, its nodes will be colored consecutively from $0$ to $2k-1$ in a run of the loop. Indeed, at each iteration of the for-loop, every node $u\in V$ picks a color $c(u)\in \{0,\dots,2k-1\}$ uniformly at random, at Instruction~\ref{inst:coloring}. Provided with this coloring~$c$, the algorithm proceeds as a sequence of three different color-BFS calls. 

The first color-BFS (see Instruction~\ref{inst:colorbfs-U}) aims at detecting the presence of a $2k$-cycle containing  light nodes only. Therefore, $X$ is set to~$U$,  the graph~$H$ is merely the subgraph $G[U]$ of~$G$ induced by~$U$, and the threshold is~$\tau$. Then, since $G[U]$ has maximum degree~$n^{1/k}$, as we shall see in Lemma~\ref{lemma:success_U}, the threshold cannot be exceeded. As a consequence, thanks to Fact~\ref{fact:link_I_X}, if the coloring~$c$ does color the nodes of a $2k$-cycle consecutively from~$0$ to~$2k-1$, then this cycle will be detected. 

The second color-BFS (see Instruction~\ref{inst:colorbfs-S}) aims at detecting the presence of a $2k$-cycle containing at least one selected node, i.e., at least one node in~$S$. Therefore, $X$ is set to~$S$, and the graph~$H$ is the whole graph~$G$. The threshold is again set to~$\tau$, which exceeds the expected size of~$S$. Therefore, thanks to Fact~\ref{fact:link_I_X}, if there is a $2k$-cycle including a node in~$S$, this cycle will be detected. 

The third color-BFS (see Instruction~\ref{inst:colorbfs-W})  is addressing the ``general case'', that is the detection of a cycle containing at least one heavy node (i.e., with at least one node not in~$U$), and not containing any selected nodes (i.e., with no nodes in~$S$). The search is therefore performed in the subgraph $G[V\smallsetminus S]$ of~$G$. The crucial point in the third color-BFS is the choice of the set $X$ initializing the search. We set it to the aforementioned set $W = \{u\in V\smallsetminus S \mid |N_G(u)\cap S| \geq k^2\}$. That is, the search is activated only by nodes neighboring sufficiently many, i.e., at least $k^2$, selected nodes. As for the first two color-BFS calls, the threshold is set to~$\tau$. This may appear counter intuitive as $W$ may be larger than $O(n^{1-1/k})$. Nevertheless, we shall show that if a node $v\in G[V\smallsetminus S]$ has to forward more than $O(n^{1-1/k})$ identifiers of nodes in~$W$, then there must exist a $2k$-cycle in~$G$ passing through $S$, which would have been detected at the previous step.

Finally, each node having performed $K$ iterations of the for-loop, without rejecting during any execution of the $3K$ color-BFS procedures, accepts. This completes the description of the algorithm. 

\subsection{Analysis of Algorithm~\ref{algo-C2k-freeness}}

\withoutappendix{Let us start the analysis of Algorithm~\ref{algo-C2k-freeness} by a collection of basic technical facts.}

\subsubsection{Preliminary Results}

\begin{toappendix}
\withappendix{\subsection{Statistical facts}}
We first show that when a $2k$-cycle exists, by setting the constant $K$ large enough, there will be a run of the loop at Instruction~\ref{inst:nb_colorings} of Algorithm~\ref{algo-C2k-freeness} for which the nodes of the cycle are colored consecutively from~0 to $2k-1$ by the coloring~$c$.

\begin{fact}\label{fact:good_colors}
Let $\alpha>0$, and
$K\geq \alpha (2k)^{2k}$.
Let $(u_0,u_1,\dots,u_{2k-1})$ be  any sequence of $2k$ nodes of $G$.
Then, with probability at least $1-e^{-\alpha}$,
there is at least one iteration of the loop at
Instruction~\ref{inst:nb_colorings} in Algorithm~\ref{algo-C2k-freeness}
such that $c(u_i)=i$, for $i=0,\dots,2k-1$.
\end{fact}

\begin{proof}
    For each iteration, the probability that this event occurs  is $(\nicefrac{1}{2k})^{2k}$, due to the independent choices of the colors of the nodes.
    The choices of coloring being also independent at each  iteration of the loop,  the probability that the event never occurs for any of the $K$ iterations is $$\left(1-\frac{1}{(2k)^{2k}}\right)^K\leq \exp\left({-\frac{K}{(2k)^{2k}}}\right)
    \leq e^{-\alpha},$$
    as claimed. 
\end{proof}

We carry on with two standard claims stating that, with high probability, the set $S$ is not too large, and  vertices with large degree (in $G$) have sufficiently many neighbors in~$S$.

\begin{fact}\label{fact:size_S}
Let $\alpha > 0$.
If every node selects itself with probability $p=\alpha/n^{1/k}$ at Instruction~\ref{inst:rdmS} of Algorithm~\ref{algo-C2k-freeness}, then 
    \[
    \Pr[|S|\leq 2\alpha n^{1-1/k}]\geq 1- e^{-\alpha/3}.
    \]
\end{fact}

\begin{proof}
The cardinality of $S$ is a random variable that follows a binomial distribution of parameters $(n,\nicefrac{\alpha}{n^{1/k}})$.
We  use the following Chernoff Bound: for any $\delta\geq 0$, 
$$\Pr\left[|S|\geq (1+\delta)\E[|S|]\right]\leq 
\exp\left(-\frac{\delta^2 \E[|S|]}{2+\delta}\right).$$
With $\delta=1$ and  $\E[|S|]=\alpha n^{1-1/k}$,
we conclude:
$$\Pr\left[|S|\geq 2\alpha n^{1-1/k}]\right]\leq 
\exp\left(-\frac{\alpha n^{1-1/k}}{3}\right)
\leq e^{-\alpha/3}, $$
as claimed.
\end{proof}

\begin{fact}\label{fact:deg_S_u_0}
Let  $\alpha>0$, and let $v\in V$ such that $\deg(v) > n^{\nicefrac{1}{k}}$.
If every node selects itself with probability $p=\nicefrac{\alpha}{n^{1/k}}$ at Instruction~\ref{inst:rdmS} of Algorithm~\ref{algo-C2k-freeness}, then
    \[
    \Pr[|N_G(v)\cap S| \geq \alpha /2]\geq 1-e^{-{\alpha}/{8}}.
    \]
\end{fact}

\begin{proof}
The degree $|N_G(v)\cap S|$ of $v$ in $S$ follows a binomial law of parameters 
 $(\deg(v),p)$.  
We  use the following Chernoff Bound: for any $0<\delta<1$, 
$$\Pr\left[|N_G(v)\cap S|\leq (1-\delta)\E[|N_G(v)\cap S|]\right]\leq 
\exp\left(-\frac{\delta^2 \E[|N_G(v)\cap S|]}{2}\right).$$
For $\delta=1/2$, we have $\E[|N_G(v)\cap S|]=p\deg(v)\geq \alpha$, from which it follows that
$$\Pr\left[|N_G(v)\cap S|\leq \alpha/2\right]\leq
\Pr\left[|N_G(v)\cap S|\leq p\deg(v)/2\right]\leq 
\exp\left(-\frac{ p\deg(v)}{8}\right)
\leq 
e^{-{\alpha}/{8}},
$$
as desired.
\end{proof}
\end{toappendix}

Let us \withoutappendix{now}\withappendix{first} introduce some notations used for analysing  Procedure \textsf{color-BFS}($k,H,c,X,\tau$) in 
Algorithm~\ref{algo-C2k-freeness}.
Let $$X_0=\{x\in X \mid c(x)=0\}$$ be
the set of nodes that send their identifiers at Instruction~\ref{ins:BFS-initialization} of $\textsf{color-BFS}(k,H,c,X,\tau)$.
We define, for any node $v$ in $H$ colored $i$ or $2k-i$ with $i\in\{1,\dots,k-1\}$, 
the subset $X_0(v)\subseteq X_0$ of nodes connected to $v$ by a ``well colored'' path in $H$, as follows.
\begin{itemize}    
    \item If $c(v)\in\{1,\dots,k-1\}$ then 
    \begin{equation}\label{eq:X_0_small_color}
    \begin{split}
        X_0(v)=   \{ & x\in X_0 \mid  \mbox{there exists a path 
     $(x,v_1,\dots,v_{c(v)-1},v)$ in $H$ such that,}  \\
     &\mbox{for every $j\in\{1,\dots,c(v)-1\}$, $c(v_j)=j$}\}.
    \end{split}
    \end{equation}

    \item If $c(v)\in\{2k-1,\dots,k+1\}$ then 
    \begin{equation}\label{eq:X_0_big_color}
    \begin{split} X_0(v)=\{ & x\in X_0 \mid \mbox{there exists a path $(x,v_{2k-1},\dots,v_{c(v)+1},v)$ in $H$ such that}  \\
     &\mbox{for every $j\in\{c(v)-1,\ldots,2k-1\}$, $c(v_j)=j$}\}.
     \end{split}\end{equation}
\end{itemize}
With this construction in hand, we can state a first general fact for any instance of \textsf{color-BFS}($k,H,c,X,\tau$)
from Algorithm~\ref{algo-C2k-freeness}.

\begin{fact}\label{fact:link_I_X}
Procedure \textsf{color-BFS}($k,H,c,X,\tau$) satisfies the following two properties.
\begin{itemize}
    \item  Let $(v_1,v_2,\ldots,v_{k-1})$ be a path in $H$ such that $c(v_i)=i$ for all $i\in\{1,\dots,k-1\}$. If $|X_0(v_{k-1})|\leq\tau$, then, for every $i\in\{1,\dots,k-1\}$, $|I_{v_i}|\leq \tau$ and $I_{v_i} = X_0(v_i)$. 

    \item Let $(v_{2k-1},v_{2k-2},\ldots,v_{k+1})$ be a path in $H$ such that $c(v_i)=i$ for all $i\in\{k+1,\dots,2k-1\}$. If $|X_0(v_{k+1})|\leq\tau$, then, for every $i\in \{k+1,\dots,2k-1\}$, $|I_{v_i}|\leq \tau$ and  $I_{v_i} = X_0(v_i)$.
\end{itemize}
\end{fact}

\begin{proof}
We prove the result  for nodes with colors between 1 and $k-1$ only, as the result for nodes with colors between $2k-1$ and $k+1$ holds by the same arguments.
From Instructions~\ref{ins:BFS-initialization} and~\ref{ins:Iv},  we get that for every node $v$ in $H$, every identifier in~$I_v$ is the one of some node in $X_0(v)$. All along a path $(v_1,\dots,v_{k-1})$, if $|X_0(v_i)|\leq \tau$ for every $i\in\{1,\dots,k-1\}$, then the condition of Instruction~\ref{ins:I_size} is always satisfied for $I_{v_i}$, and the identifiers of all the nodes $X_0(v_i)$ are in $I_{v_i}$. The result follows after noticing that  $X_0(v_i)\subseteq X_0(v_{k-1})$ for every $i\in\{1,\dots,k-1\}$.
\end{proof}

\subsubsection{Analysis of the first two $\textsf{color-BFS}$} 

The next lemma states that, if $G$ contains a $2k$-cycle composed 
of consecutively colored nodes which are all light, that is, with small degree, then the first call of $\textsf{color-BFS}$ leads at least one node to reject.

\begin{lemma}\label{lemma:success_U}
    Suppose that $G$ contains a $2k$-cycle $C=(u_0,\dots,u_{2k-1})$ in $G[U]$, and that $c$ is a coloring for which $c(u_i)=i$ for every $i\in\{0,\dots,2k-1\}$. 
    Then $u_k$ rejects in $\textsf{color-BFS}(k,G[U],c,U,\tau)$.
\end{lemma}

\begin{proof}
    Every node in $U$ has degree at most~$n^{1/k}$. By a direct induction on $i=1,\dots, k-1$, we get that $|U_0(u_i)|\leq n^{\nicefrac{i}{k}}\leq\tau$ for every $i\in\{0,\dots,k-1\}$, and the same holds for $u_{2k-i}$. Fact~\ref{fact:link_I_X} implies that $|I_{u_i}|\leq\tau$ and $|I_{u_{2k-i}}|\leq\tau$ for every $i\in\{1,\dots,k-1\}$. Therefore the identifier of $u_0$ is forwarded along the two paths $(u_0,u_1,\dots,u_{k-1},u_k)$ and $(u_0,u_{2k-1},\dots,u_{k+1},u_k)$, and $u_k$ rejects.
\end{proof}

The next lemma establishes that if $S$ is of size at most $\tau$,  and if $G$ contains a well colored cycle for which the node colored~0 is in~$S$, then the second call of $\textsf{color-BFS}$ leads the node colored~$k$ in the cycle to reject.

\begin{lemma}\label{lemma:success_S}
Suppose that $G$ contains a $2k$-cycle $C=(u_0,\dots,u_{2k-1})$ with $u_0\in S$, and let $c$ be a coloring such that $c(u_i)=i$ for every $i\in\{0,\dots,2k-1\}$.  If $|S|\leq \tau$ then $u_k$ rejects in $\textsf{color-BFS}(k,G,c,S,\tau)$.
\end{lemma}

\begin{proof}
 As $S_0(v)\subseteq S$ for every node~$v\in G$, $|S_0(v)|\leq |S|$, and thus in particular for every node $v\in\{u_1,\dots,u_{k-1}\}\cup\{u_{2k-1},\dots,u_{k+1}\}$.  Fact~\ref{fact:link_I_X} then implies that $|I_{u_i}|\leq\tau$ and $|I_{u_{2k-i}}|\leq\tau$ for every $i\in\{1,\dots,k-1\}$. Therefore, the identifier of node $u_0$ is forwarded along the two branches of the cycle, and $u_k$ rejects.
\end{proof}

\subsubsection{Analysis of the third $\textsf{color-BFS}$}

The purpose of this section is to prove the following lemma which, combined with Lemmas~\ref{lemma:success_U} and~\ref{lemma:success_S}, ensures the correctness of our algorithm. Informally, Lemma~\ref{lemma:success_W} below states that if $G$ contains no cycle passing through $S$, but contains a well colored  cycle passing through some $u_0$ having sufficiently many neighbours in~$S$, then the third call to $\textsf{color-BFS}$ rejects.
This lemma is crucial as it does not only provide the key of the proof of Theorem~\ref{classical_theo}, but it is also central in the design of our algorithm for quantum \CONGEST.

\begin{lemma}\label{lemma:success_W}
Suppose that $G$ contains a $2k$-cycle $C=(u_0,\dots,u_{2k-1})$ with $c(u_i)=i$ for every $i\in \{0,\dots,2k-1\}$,
which is not intersecting~$S$, but such that $u_0$ has at least $k^2$ neighbors in $S$.  If  $|S|\leq \tau/(k2^{k-1})$ then $u_k$ rejects in $\textsf{color-BFS}(k,G[V\smallsetminus S],c,W,\tau)$.
\end{lemma}

The proof of Lemma~\ref{lemma:success_W} is based on the following combinatorial property. If $C\cap S=\varnothing$ for every $2k$-cycle $C$ in $G$, then, for every $v\in V\smallsetminus S$, the set $W_0(v)$ defined below is small, which implies that $I_v$ is also small in $\textsf{color-BFS}(k,G[V\smallsetminus S],c,W,\tau)$. More specifically, the following holds. 

\begin{lemma}[Density lemma]\label{lem:combmain}
Let $S, W_0,V_1,\dots, V_{k-1}$ be non-empty disjoint subsets of vertices of a graph~$G$, and let us assume that every vertex $w_0 \in W_0$ has at least $k^2$ neighbors in $S$. 
For every $i\in\{1,\dots,k-1\}$, and every $v \in V_i$, let us define 
\[
W_0(v) = \{w \in W_0 \mid \mbox{$G$ contains a path $(w,v_1,\dots, v_i=v)$ s.t., for every $j\in\{1,\dots,i-1\}$, $v_j \in V_j$}\}.
\]
If $|W_0(v)| > 2^{i-1}(k-1) |S|$ for some $i\in\{1,\dots,k-1\}$  and some $v\in V_i$, then $G$ contains a $2k$-cycle intersecting $S$. 
\end{lemma}

Before proving Lemma~\ref{lem:combmain}, let us observe that it directly implies Lemma~\ref{lemma:success_W}.

\begin{proof}[Proof of Lemma~\ref{lemma:success_W}]
Let us apply  Lemma~\ref{lem:combmain} twice, once with 
\[
V_i=\{v\in V\smallsetminus S\mid c(v)=i\}
\]
for every $i\in\{1,\dots,k-1\}$, and once with 
\[
V_i=\{v\in V\smallsetminus S\mid c(v)=2k-i\}.
\]
We get that, for $\textsf{color-BFS}(k,G[V\smallsetminus S],c,W,\tau)$, if $C\cap S=\varnothing$ for every $2k$-cycle $C$ in $G$ then, for every node $v\in V\smallsetminus S$ colored~$i$ or~$2k-i$ with $i\in\{1,\dots,k-1\}$, we have $|W_0(v)|\leq 2^{i-1}(k-1)|S|$. Using the fact that $|S|\leq\tau/(k2^{k-1})$, we derive that $|W_0(v)| \leq \tau $. It infers by Fact~\ref{fact:link_I_X} that $I_{u_i} = W_0(u_i)$ for all $i\in\{1,\dots, k-1\}\cup\{2k-1, \dots, k+1\}$. In particular the identifier of $u_0$ is forwarded along the two branches of the cycle, and $u_k$ rejects.
\end{proof}

The rest of this subsection is devoted to the proof of the Density Lemma (Lemma~\ref{lem:combmain}).

As a warm-up, let us prove the Density Lemma for $i = 1$, in order to provide the reader with an intuition of the proof before getting into the (considerably more complicated) general case. The case $i=1$ will also illustrate why the lemma is called ``Density lemma''. 

\paragraph{Warm Up: The Case $i=1$.} 

 Let  $v \in V_1$, and let us consider the bipartite graph $H(v)$ with vertex partition $S$ and $W_0(v)$. (Recall that $W_0(v) \subseteq W_0$.) Suppose that $|W_0(v)|>(k-1)|S|$. Let us show that one can construct a path $P$, of length $2(k-1)$, with both end points in $W_0(v)$. This path together with $v$ forms a $2k$-cycle. Path~$P$ exists merely because the graph $H(v)$ is dense. More precisely, $H(v)$ has degeneracy\footnote{A graph is $d$-degenerate if each of its subgraphs has a vertex of degree at most~$d$; The degeneracy of a graph is the smallest value~$d$ for which it is $d$-degenerate.} at least~$k$. Indeed, since every $w \in W_0$ has at least $k^2$ neighbours in $S$, we get that 
\[
        |E(H(v))|
        \geq k^2|W_0(v)| 
        \geq 2k|W_0(v)| 
        \geq k(|W_0(v)|+|S|) 
        = k|V(H(v))|.
\]
Consequently, $H(v)$ contains a non-empty subgraph of minimum degree at least~$k$, obtained by repeatedly removing all vertices of degree less than~$k$. In a bipartite graph of minimum degree~$k$, one can greedily construct a path of $2k$ vertices, starting from any vertex, and extending it $2k-1$ times with a new vertex that has not been used before. The extension is possible as long as the path has no more than $k-1$ vertices in each of the partitions. In particular, one can construct a path with $2k-1$ edges with one extremity in $S$ and one in $W_0(v)$. Therefore there exists a path $P$ with $2(k-1)$ edges, starting and ending on vertices of $W_0(v)$, 
thus achieving the proof for $i=1$. \qed

\medbreak

We aim at extending the proof above to arbitrary values of~$i$. Let $i$ be the smallest index for which the condition of the lemma, i.e., $|W_0(v)|>2^{i-1}(k-1)|S|$ for some $v \in V_i$, holds. A naive approach consists in considering again the bipartite graph $H(v)$ with edge set $E(S,W_0(v))$. As in the case~$i=1$, one may argue that $H(v)$ is dense enough to contain a path $P$ of length $2(k-i)$, starting and ending on some vertices of $W_0(v)$. Both endpoints of path $P$ have some path to $v$, say $P'$ and $P''$ respectively, of length $i$. If paths $P'$ and $P''$ were disjoint (except for $v$), then they would form, together with path $P$, a $2k$-cycle. Unfortunately, with a greedy construction of path $P$ as in the case $i=1$, there are no reasons why paths $P'$ and $P''$ should be disjoint, hence the need to refine the analysis. 

\paragraph{Intuition of the proof.}(Also see Fig.~\ref{sparsification} \withappendix{in Appendix~\ref{app-cycle}}.) Let us  construct a  sequence of nested subgraphs of $H(v)$, denoted~by 
\[
\H(v,0), \H(v,1),  \dots, \H(v,2\qGamma),
\]
where $\qGamma = \lfloor \frac{k-1}{2}\rfloor$. The construction starts from $\H(v,2\qGamma)$, down to $\H(v,0)$, with $\H(v,\gamma-1)\subseteq \H(v,\gamma)$ for every $\gamma\in\{1,\dots,2\qGamma\}$. We prove (cf. Lemma~\ref{lem:level_cycle}) that, if $\H(v,0)$ is not empty, then $G$ contains a $2k$-cycle passing through~$S$. Then we conclude by showing (cf. Lemma~\ref{lem:color_survive}) that if the graph $\H(v_j,0)$ is  empty for every $j\in\{1,\dots, i\}$ and every $v_j \in V_j$, then $|W_0(v)|\leq 2^{i-1}(k-1)|S|$.

In order to construct the decreasing sequence of graphs $\H(v,\gamma)$, $\gamma=2\qGamma,\dots,0$, each vertex $v$ ``selects'' a subset of edges of~$H(v)$, denoted by~$\F(v)$. One can think of these edge-sets as edges selected by vertex~$v$ for its neighbors of color $c(v)+1$, and vertices of color $c(v)+1$ will only be allowed to choose among edges selected for them by their neighbors of color~$c(v)$.
More specifically, if $v$ is colored $0$, then $\F(v)$ is simply the set of edges incident to $v$ connecting $v$ to a vertex of $S$. Then, for any vertex $v$ of color $i\geq1$, $v$ constructs the sequence $\H(v,\gamma)$ starting from the union of $\F(v')$ for all neighbors $v'$ of $v$ with color $i-1$. The sequence $\H(v,\gamma)$ is constructed by repeatedly removing edges incident to vertices of ``small'' degree, and $\F(v)$ corresponds to some of the edges removed during the process. 
This construction will ensure that if, for some vertex $v$ of color $i$, $\H(v,0)\neq \emptyset$, then $G$ contains a $2k$-cycle (Lemma~\ref{lem:level_cycle}), and the reason for this is that one can construct the three following paths (see Fig.~\ref{sparsification}):

\begin{enumerate}
\item a path $P$ of length $2(k-i)-1$ in $\H(v,2\qGamma)$, starting from some $w \in W_0(v)$, and ending on some $s \in S$ of high degree;

\item a path $P'=(v'_0,v'_1,\dots, v'_i-1,v'_i)$ of length $i$ from the endpoint $w=v'_0$ of $P$ to $v=v'_i$, such that, for every $j <i$,  $v'_j$ is colored~$j$, and the edge incident to $w$ in $P$ is contained in the graph $\F(v'_j)$;

\item a path $P''=(s,v''_0,\dots, v''_i)$ of length $i+1$ from the other endpoint $s$ of $P$ to $v=v''_i$, where each $v''_j$ is colored $j$, such that $P''$ intersects $P$ only in $s$, and intersects $P'$ only in~$v$. For constructing this third path $P''$, we heavily rely on the maximum degrees of the graphs $\F(v_j)$ (which are upper bounded by a function of $j$), and on the fact that the edge of $P''$ incident to $s$ can be chosen to avoid all sets $\F(v_j)$, for all $j<i$. This will ensure that $P''$ and $P'$ are disjoint.

\end{enumerate}
The paths $P$, $P'$, and $P''$ together form the cycle of length~$2k$. Provided with this rough intuition, let us proceed with the formal construction of the graphs $\H$ and $\F$.  

\paragraph{Sparsification.}

Let $E(S,W_0)$ denote the set of edges of $G$ having one endpoint in $S$, and the other in~$W_0$. Let $\F(v),\H(v)\subseteq E(S,W_0)$ be defined inductively on $i=1,\dots,k-1$ for any node $v\in V_i$.  $\F(v)$ and $\H(v)$ are constructed as sets of edges, and we slightly abuse notation by also denoting $\F(v)$  and $\H(v)$ as the \emph{graphs} induced by these edge sets.

For every $w\in W_0$, $\F(w)$ is 
defined as the set of edges between $w$ and $S$ in $G=(V,E)$, that is, for every $w\in W_0$,
\begin{equation}\label{eq:F_col0}
\F(w)=E(\{w\},S)=\big \{\{w,s\}\in E \mid s\in S\big\}. 
\end{equation}
Let us set $V_0=W_0$. Then, for every $i\in\{1,\dots,k-1\}$, and every 
$v\in V_i$, we inductively define  
\begin{equation}
\label{eq:init_F}
\H(v)=\bigcup_{\{v,v'\}\in E \; : \; v'\in V_{i-1}} \F(v').
\end{equation}
Let $\qGamma=\big\lfloor\frac{k-i}{2}\big\rfloor$. We now inductively define 
an intermediate increasing sequence of subsets 
\[
\H(v,0) \subseteq \H(v,1)\subseteq \ldots\subseteq \H(v,2\qGamma-1) \subseteq \H(v,2\qGamma)\subseteq \H(v)
\]
as follows:
\begin{enumerate}
\item Initialization: 
\begin{equation}\label{init_inGamma}
\H(v,2\qGamma)=\big\{\{s,v\}\in \H(v) \mid (s\in S) \;\land\; \big(\deg_{\H(v)}(s)> 2^{i-1}(k-1)\big)\big\}.
\end{equation}
\item Induction from $2\gamma$ to $2\gamma-1$:
\begin{equation}\label{even_to_odd}
\H(v,2\gamma-1)=\big \{\{w,s\}\in \H(v,2\gamma) \mid (w\in W_0)\;\land\; \big(\deg_{\H(v,2\gamma)}(w)> 2\gamma\big)\big \}.
\end{equation}
\item Induction from $2\gamma-1$ to $2\gamma-2$: 
\begin{equation}\label{odd_to_even}
\H(v,2\gamma-2)= \big\{\{s,w\}\in \H(v,2\gamma-1) \mid (s\in S)\;\land\; \big(\deg_{\H(v,2\gamma-1)}(s)> 2\gamma-1\big)\big \}.
\end{equation}
\end{enumerate}
Finally, we define $\F(v)$ as:
\begin{equation}
\label{eq:init_out}
\begin{split}
    \F(v)
    =& \Big\{\{s,v\}\in \H(v) \mid (s\in S) \;\land\; \big(\deg_{\H(v)}(s)\leq 2^{i-1}(k-1)\big)\Big\}\\
    &\bigcup_{\gamma=1}^\qGamma \Big\{\{s,v\}\in \H(v,2\gamma-1) \mid (s\in S) \;\land\; \big(\deg_{\H(v,2\gamma-1)}(s)\leq 2\gamma -1 )\big)\Big\}.
\end{split}
\end{equation}

\withoutappendix{This construction is central because if $\H(v,0)\neq \varnothing$ then there is a $2k$-cycle in the graph, passing through $S$ (Lemma~\ref{lem:level_cycle}).}\withappendix{This construction has properties (cf. Appendix~\ref{app-prop}) leading to the proof of Lemma~\ref{lem:combmain} in two steps.}
\begin{toappendix}
\withappendix{\subsection{Properties of $\H(v)$ and $\F(v)$}\label{app-prop}}%
Let us first state two preliminary results on $\H(v)$ and $\F(v)$.

\begin{fact}\label{fact:degmax_OUT}
Let $i\in\{1,\dots,k-1\}$, and let $v\in V_i$. For any $s\in S\cap\H(v)$, we have $$\deg_{\F(v)}(s)\leq 2^{i-1}(k-1).$$
\end{fact}

\begin{proof}
By construction thanks to Eq.~\eqref{eq:init_out}, and by the fact that $2\qGamma-1\leq k-2\leq2^{i-1}(k-1)$.
\end{proof}

\begin{lemma}\label{lem:path_F}
Let $i\in\{1,\dots,k-1\}$, and let $v\in V_i$.
For every edge $\{s,w\}\in \H(v)$, with $s\in S$ and $w\in  W_0$, there is a path $(w,v_1,\dots,v_{i-1},v)$ such that, for every $j\in\{1,\dots,i-1\}$, $v_j\in V_j$ and $\{s,w\}\in\F(v_j)$. 
In particular, $\H(v)\subseteq E(S,W_0(v))$.
\end{lemma}

\begin{proof}
The second claim, $\H(v)\subseteq E(S,W_0(v))$, is a direct consequence of the first. The first statement is proven by induction on $i\geq 1$. For $i=1$, i.e., for $v\in V\smallsetminus S$ with $c(v)=1$, and for $\{s,w\}\in\H(v)$, it follows from Eq.~\eqref{eq:init_F} that there exists $v'\in V\smallsetminus S$ such that $\{v,v'\}\in E$, $c(v')=0$, and $\{s,w\}\in \F(v')$. By construction, $\F(v')=E(\{v'\},S)$ when $c(v')=0$ (see Eq.~\eqref{eq:F_col0}), therefore $w=v'$, and the required path is just the edge $\{w,v\}$.

    Assume now that the lemma holds for color $i\geq 1$. Fix $v\in V\smallsetminus S$ such that $c(v)=i+1$. Again, by Eq.~\eqref{eq:init_F}, if $\{s,w\}\in\H(v)$, then there exists $v'\in V\smallsetminus S$ such that $\{v,v'\}\in E$, $c(v')=i$, and $\{s,w\}\in \F(v')$. By construction (see Eq.~\eqref{eq:init_out}), $\F(v')\subseteq\H(v')$. As a consequence,  we can apply the induction hypothesis on $\{s,w\}\in \H(v')$, which provides us with a path $(w,v_1,\dots,v_{i-1},v')$. We further extend this path using the edge $\{v,v'\}\in G[V\smallsetminus S]$, which concludes the induction.
\end{proof}
\end{toappendix}

\paragraph{Cycle construction.}

The following lemma \withappendix{(see proof in Appendix~\ref{app-cycle})} implies that if there exists $v\in V\smallsetminus S$ such that $\H(v,0)\neq\varnothing$, then there is a $2k$-cycle intersecting~$S$. Such a cycle is exhibited in Figure~\ref{sparsification} for $k=5$ and $i=2$. 
\withoutappendix{Associated with Lemma~\ref{lem:color_survive} that provides the existence of such a $v$, it will prove Lemma~\ref{lem:combmain}.}

\begin{lemma}
\label{lem:level_cycle}
Assume that $\H(v,0)\neq \varnothing$ for some $v\in V\smallsetminus S$ with $c(v)=i\in\{1,\dots,k-1\}$.
Then there is a $2k$-cycle $C$ in $G$ such that $C\cap S\neq \varnothing$.
\end{lemma}

\begin{toappendix}
\withappendix{\subsection{Proof of Lemma~\ref{lem:level_cycle}}\label{app-cycle}}
\begin{proof}
The existence of  a $2k$-cycle in $G$ going through~$S$ is shown by constructing three simple paths $P$ (Claim~\ref{claim:path_S_W}), $P'$  and $P''$ (Claim~\ref{claim:vertical_paths}), the union of which is a cycle. In Figure~\ref{sparsification}, path $P=(w,s_3,w_2,s_1,w'_2,s)$, path $P'=(w,v'_1,v)$ and path $P''=(s,w'',v''_1,v)$.

\begin{figure}[ht]
    \centering  \includegraphics[width=9cm]{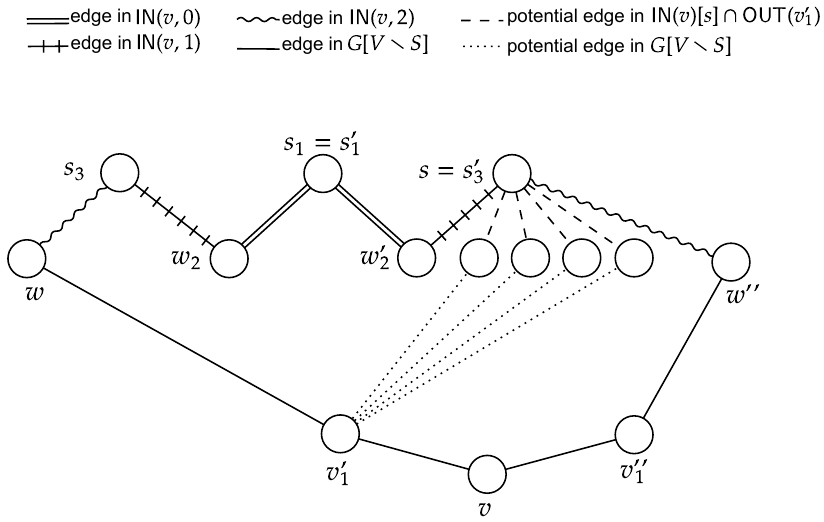}
    \caption{The case of a $10$-cycle (i.e., $k=5$). In the figure, $\H(v,0)\neq\varnothing$ as $v\in V_2$.
    Here we have $\qGamma=1$, and thus the considered graphs for the proof are  $\H(v,0)\subseteq\H(v,1)\subseteq\H(v,2)\subseteq\H(v)$. Regarding the proof of Claim~\ref{claim:path_S_W}, we have $\deg_{\H(v,1)}(s_1)>1$, and thus there exist vertices $w_2$ and $w'_2$ in $\H(v,1)$. Since $\deg_{\H(v,2)}(w_2)>2$, and since $\deg_{\H(v,2)}(w'_2)>2$, there are two vertices $s_3$ and $s'_3$ in $\H(v,2)$. And since $\deg_{\H(v)}(s_3)>8$, there exists a vertex $w$ in $\H(v)$.
    Regarding the proof of Claim~\ref{claim:vertical_paths}, we have $\deg_{\H(v)}(s)>8$ and $\deg_{\F(v'_1)}(s'_3)\leq 4$. Therefore, there exists a vertex $w''$ in $\H(v)[s]\smallsetminus\big(\{w,w_2,w'_2\}\cup\F(v'_1)\big)$.}
    \label{sparsification}
\end{figure}

\begin{claim}\label{claim:path_S_W}
There is a simple path $P$ with $2(k-i)$ nodes in $W_0\cup S$, alternating between nodes in $W_0$ and nodes in $S$, the edges of which are all in $\H(v,2\qGamma)$.
\end{claim}

\begin{claimproof}
The construction  is done by induction on $\gamma$. 
A path
 \[
 P_{\gamma}=(s_{2\gamma+1},\dots,w_2,s_1,w'_2,\dots,s'_{2\gamma+1})\subseteq \H(v,2\gamma)
 \]
is aimed at being extended to a path $P_{\gamma+1}\subseteq \H(v,2\gamma+2)$. 
 The extension uses the fact that both extremities of $P_{\gamma}$ have some incident edges in $\H(v,2\gamma)$. 
    
For the base case $\gamma=0$, let $s_1=s'_1\in S$ be any node with an  incident edge in $\H(v,0)\neq\varnothing$. Observe that, in this base case, the path is trivial, formed by a unique vertex of graph $\H(v,0)$.
    
For the induction step, we start from a path 
$P_{\gamma}=(s_{2\gamma+1},\dots,s_1,\dots,s'_{2\gamma+1})$ 
in $\H(v,2\gamma)$, from $S$ to $S$. Note that this path is also a path in $\H(v,2\gamma+2)\supseteq\H(v,2\gamma)$.
Moreover, since $s_{2\gamma+1}$ has some incident edge in $\H(v,2\gamma)$, and since $\H(v,2\gamma)\subseteq\H(v,2\gamma+1)$,
$s_{2\gamma+1}$ must have large degree in $\H(v,2\gamma+1)$. Specifically, thanks to Eq.~\eqref{odd_to_even}, we have
\[
\deg_{\H(v,2\gamma+1)}(s_{2\gamma+1})= \deg_{\H(v,2\gamma)}(s_{2\gamma})>2\gamma+1.
\]
Thus, there is a vertex $w_{2\gamma+2}\in W_0$ such that 
\[
w_{2\gamma+2}\not \in\{w_2,w'_2,\dots,w_{2\gamma},w'_{2\gamma}\}
\;\mbox{and}\; 
\{w_{2\gamma+2},s_{2\gamma+1}\}\in \H(v,2\gamma+1).
\]
Similarly, we can find  $w'_{2\gamma+2}$ adjacent to $s'_{2\gamma+1}$ in $\H(v,2\gamma+1)$ (see edges $(s_1,w_2)$ and $(s'_1,w'_2)$ in Figure~\ref{sparsification}). The degrees of $w_{2\gamma+2}$ and $ w'_{2\gamma+2}$ in $\H(v,2\gamma+2)$ are equal to their degrees in $\H(v,2\gamma+1)$, and thus they are both at least $2\gamma+2$ (by Eq~\eqref{even_to_odd}). 
This establishes the existence of two new vertices $s_{2\gamma+3}$ and $s'_{2\gamma+3}$ (see edges $(w_2,s_3)$ and $(w'_2,s'_3)$ in Figure~\ref{sparsification}).
The extended path is then
\[
P_{\gamma+1}=(s_{2\gamma+3},w_{2\gamma+2},s_{2\gamma+1},\dots,s_1,\dots,s'_{2\gamma+1},w'_{2\gamma+2},s'_{2\gamma+3}),
\]
which concludes the induction step. The path $P_\qGamma$ has $4\qGamma+1$ vertices, and has its end points in~$S$ (path $P_1=(s_3,w_2,s_1,w'_2,s'_3)$ in Figure~\ref{sparsification}). 

If $k-i$ is even, then $4\qGamma+1=2(k-i)+1$. So the path $P$ is merely $P_\qGamma$ without one of its two end points. 

If $k-i$ is odd (as in Figure~\ref{sparsification}), then the path $P_\qGamma$ has  $2(k-i)-1$ vertices, and we need one more step to reach the desired length $2(k-i)$. 
Since $\{s_{2\qGamma+1},w_{2\qGamma}\}\in\H(v,2\qGamma)$, by using Eq.~\eqref{init_inGamma}, we get $$\deg_{\H(v)}(s_{2\qGamma+1})=\deg_{\H(v,2\qGamma)}(s_{2\qGamma+1})>2^{i-1}(k-1)>2\qGamma=k-i-1.$$ 
Therefore we can select $w\notin\{w_2,w'_2,\dots,w_{2\qGamma},w'_{2\qGamma}\}$ such that $\{s_{2\qGamma+1},w\}\in\H(v,2\qGamma)$ (see edge $(s_3,w)$ in Figure~\ref{sparsification}). Then $P$ is  $P_\qGamma$ augmented with $w$ connected to its
end point~$s_{2\qGamma+1}$.
\end{claimproof}

Let us now connect the two extremities of $P$ to node $v$ by using two well-colored node-disjoint paths. 

\begin{claim}\label{claim:vertical_paths}
There are two simple paths:
\begin{enumerate}
\item  $P'=(w,v'_1,\dots,v'_{i-1},v)$, 
with $i+1$ nodes from the extremity $w\in W_0$ of~$P$ to $v$,
such that, for every $j\in\{1,\dots,i-1\}$, $v'_j\not\in P$, and $c(v'_j)=j$.
\item  $P''=(s,w'',v''_1,\dots,v''_{i-1},v)$, with $i+2$ nodes from the extremity $s\in S$ 
of $P$ to $v$, where $w''\in\H(v)\smallsetminus\big( P\cup P'\big)$, and, for every $j\in\{1,\dots,i-1\}$,
    $v''_j\not\in  P\cup P'$, and $c(v''_j)=j$.
\end{enumerate}
\end{claim}

\begin{claimproof}
We first construct the path $P'$. Since the edge in $P$ incident to $w$ belongs to $\H(v,2\qGamma)\subseteq\H(v)$, Lemma~\ref{lem:path_F} directly gives us a path $(w,v'_1,\dots,v'_{i-1},v)$ in $G[V\smallsetminus S]$ such that $v'_j \in V_j$ for all $j\in\{1,\dots,i-1\}$ ($P'=(w,v'_1,v)$ in Figure~\ref{sparsification}). This path $P'$  intersects $P$ only in $w$ because the nodes of $P$ are contained in $W_0 \cup S$, and the only intersection of $P'$ with this set is~$w$. 

Let us now construct the second path $P''$. Let $s$ be the extremity of $P$ in $S$. Let us denote by $\H(v)[s]$ the set of edges incident to $s$ in $\H(v)$. On the one hand, by Claim~\ref{claim:path_S_W}, the edge of $P$ incident to $s$ is in $\H(v,2\qGamma)$. By Eq.~\eqref{init_inGamma}, the degree of $s$ in $\H(v,2\qGamma)$ is larger than $2^{i-1}(k-1)$. Since $\H(v,2\qGamma) \subseteq \H(v)$, we have 
\[
\big|\H(v)[s]\big|=\deg_{\H(v)}(s)>2^{i-1}(k-1).
\]   
    On the other hand, 
\[
\Big |\H(v)[s] \cap \Big (P\cup \big(\cup_{j=1}^{i-1} \F(v'_j)\big) \Big) \Big|
\leq \big|W_0\cap P\big|+ \sum_{j=1}^{i-1}\deg_{\F(v'_j)}(s).
\]    
Moreover, thanks to Fact~\ref{fact:degmax_OUT}, $\deg_{\F(v'_j)}(s)\leq 2^{j-1}(k-1)$ for every $j\in\{1,\dots,i-1\}$. As a consequence,  
\[
\Big |\H(v)[s] \cap \Big (P\cup \big(\cup_{j=1}^{i-1} \F(v'_j)\big) \Big) \Big|
\leq  \ k-i+\sum_{j=1}^{i-1}2^{j-1}(k-1)
 \leq 2^{i-1}(k-1).
\]
It follows that 
\[
\Big |\H(v)[s] \cap \Big (P \cup \big(\cup_{j=1}^{i-1} \F(v'_j)\big) \Big) \Big|  <\big|\H(v)[s]\big|.
\]
Therefore, there exists an edge $e=\{s,w''\}\in \H(v)$ which is neither in~$P$, nor in any $\F(v'_j)$ for every $j\in\{1,\dots,i-1\}$ (see Figure~\ref{sparsification}).
Finally, by Lemma~\ref{lem:path_F}, there exists a path $(w'',v''_1,\dots,v''_{i-1},v)$ in $G[V\smallsetminus S]$ such that, for every $j\in\{1,\dots,i-1\}$, $c(v''_j)=j$, and $e\in\F(v''_j)$. Since $e\notin \F(v'_j)$ for any $j\in\{1,\dots,i-1\}$, and since nodes in $P$ are either in $S$, or colored~$0$, none of the nodes $v''_1,\dots,v''_{i-1}$ are in $P\cup P'$, which completes the proof of the claim. 
\end{claimproof}

By construction, the paths $P, P'$, and $P''$ from Claims~\ref{claim:path_S_W}-\ref{claim:vertical_paths} satisfy that $P\cup P'\cup P''$ is a cycle. By denoting $|Q|$ the number of vertices of a path~$Q$, the  length of $P\cup P'\cup P''$ is $$(|P|-1)+(|P'|-1)+(|P''|-1)=(2(k-i)-1)+((i+1)-1)+((i+2)-1)=2k.$$
which completes the proof of Lemma~\ref{lem:level_cycle}. 
\end{proof}
\end{toappendix}

\paragraph{Congestion.} 
\withoutappendix{To complete the proof of Lemma~\ref{lem:combmain}, 
we establish a lemma 
that}\withappendix{The following lemma (see proof in Appendix~\ref{app-color})} allows us to bound the size of $W_0(v)$
as a function of the size of $\F(v)$ whenever all sets $\H(v)$ are empty. 
\withoutappendix{Its contrapositive provides the hypothesis needed to apply Lemma~\ref{lem:level_cycle}, and prove Lemma~\ref{lem:combmain}.}

\begin{lemma}
\label{lem:color_survive}
Let $i\in \{1,\dots,k-1\}$. Assume that, for every $j\in\{1,\dots,i\}$, every $v_j\in V_j$ satisfies $\H(v,0)= \varnothing$. Then, for every
$v\in V_i$, 
$|W_0(v)|\leq 2^{i-1}(k-1)|S|$.
\end{lemma}
\begin{toappendix}
\withappendix{\subsection{Proof of Lemma~\ref{lem:color_survive}}\label{app-color}}
\begin{proof}
We first show that, for every $w\in W_0(v)$, there is an edge incident to~$w$ in $\F(v)$. In other words, we show that $w$ is of degree at least~1  in $\F(v)$. We first prove the more general inequality 
\[
\deg_{\F(v)}(w)\geq k^2-2i\,\qGamma,
\]
by induction on the index $i\in\{0,\dots,k-1\}$ of set $V_i$ containing node~$v$. In fact, this inequality also holds for the case $v\in W_0$, which constitutes the base case “$i=0$" of our induction, i.e., $W_0(v)=\{v\}$. In this case, the only scenario is $w=v$. Recall that  $\F(w)=E(\{w\},S)$ (see Eq.~\eqref{eq:F_col0}). Since $w\in W_0$ has at least $k^2$ neighbors in $S$, we have $\deg_{\F(w)}(w)\geq k^2$, as desired. 

Suppose now that the inequality holds for vertices colored $i-1$, and let $v\in V_i$ and $w\in W_0(v)$. By definition of $W_0(v)$, there exists a neighbor~$v'$ of~$v$ colored~$i-1$ such that $w\in W_0(v')$. By induction, 
\[
\deg_{\F(v')}(w)\geq k^2-2(i-1)\,\qGamma.
\]
As $\F(v')\subseteq \H(v)$, it also holds that 
\[
\deg_{\H(v)}(w)\geq k^2-2(i-1)\,\qGamma.
\]
We are going to estimate the maximal number of edges incident to $w$ that can be removed, i.e., that appear in $\H(v)$ but not in $\F(v)$. By construction of $\F(v)$ and $\H(v,0)$, every edge $\{s,w\}\in \H(v)$ satisfies exactly one of the three following statements:
        \begin{enumerate}
            \item $\{s,w\}\in \H(v,0)$;
            \item There exists $\gamma\in\{1,\dots,\qGamma\}$ such that $\{s,w\}\in \H(v,2\gamma)\smallsetminus\H(v,2\gamma-1)$; 
            \item $\{s,w\}\in\F(v)$.
        \end{enumerate}
By our assumption, $\H(v,0)= \varnothing$, excluding case~1.       Let us now consider case~2, if it occurs.
        Let $\gamma \in \{1,\dots,\qGamma\}$ be the largest integer
        such that $\H(v,2\gamma)\smallsetminus\H(v,2\gamma-1)$ contains some edge incident to $w$.
        Fix any $\{s,w\}\in \H(v,2\gamma)\smallsetminus\H(v,2\gamma-1)$.
Since $\{s,w\}\not\in\H(v,2\gamma-1)$, we have  
        $
        \deg_{\H(v,2\gamma)}(w)\leq 2\gamma
        $ (cf. Eq.~\ref{even_to_odd}),
        and none of the edges incident to $w$ in $\H(v,2\gamma)$
        are kept in $\H(v,2\gamma-1)$. 
        In other words, vertex~$w$ does not appear in graph $\H(v,2\gamma-1)$, and therefore it does not appear in any of the graphs $\H(v,\alpha)$, for any $\alpha \leq 2\gamma-1$. That is, all edges incident to $w$ that are in case~2 have been suppressed simultaneously, they all belong to $\H(v,2\gamma)\smallsetminus\H(v,2\gamma-1)$. 
        Therefore, we conclude the induction step as follows:
        \[
        \begin{split}
            \deg_{\F(v)}(w)
            &=\deg_{\H(v)}(w)- \deg_{\H(v,2\gamma)}(w)\\
            &\geq \big [k^2-2(i-1)\qGamma \big] - 2\gamma\\
            &\geq k^2-2i\qGamma.
        \end{split}
       \]
    Using this inequality we deduce that  
    \[
    \deg_{\F(v)}(w)\geq k^2-2i\,\qGamma \geq k^2-i(k-i)\geq k^2-(k-1)^2>0.
    \] 
Thus every vertex $w \in W_0$ has at least one incident edge in $\F(v)$, which implies that 
 $|W_0(v)|\leq |\F(v)|$. By Fact~\ref{fact:degmax_OUT}, in graph $\F(v)$, each vertex  contained in $S$ has degree at most $2^{k-2}(k-1)$, hence $|\F(v)| \leq 2^{k-2}(k-1) |S|$,
    which proves our lemma.
\end{proof}
\end{toappendix}

\begin{proof}[Proof of Lemma~\ref{lem:combmain}]
Applying first the contrapositive of Lemma~\ref{lem:color_survive}, Lemma~\ref{lem:combmain} directly follows from Lemma~\ref{lem:level_cycle}.
\end{proof}
\withappendix{
The correctness of Algorithm~\ref{algo-C2k-freeness}, and thus of Theorem~\ref{classical_theo} as well, is a consequence of Lemmas~\ref{lemma:success_U}, \ref{lemma:success_S}, and~\ref{lemma:success_W}, which are establishing that the three calls to $\textsf{color-BFS}$ decide $C_{2k}$-freeness  with small error probability, in $O(n^{1-1/k})$ rounds in \CONGEST\ (see Appendix~\ref{app-thm1} for full details). }
\begin{toappendix}
\withappendix{\subsection{Complete analysis of Algorithm~\ref{algo-C2k-freeness} (Theorem~\ref{classical_theo})}\label{app-thm1}}
\begin{proof}[Proof of Theorem~\ref{classical_theo}]
Let us first analyze the round complexity of Algorithm~\ref{algo-C2k-freeness}, and then prove that this algorithm accepts and rejects with the specified success probabilities. 

\paragraph{Complexity.}

By construction, the threshold $\tau$ ensures that the complexities of the three calls $\textsf{color-BFS}(k,G[U],c,U,\tau)$, $\textsf{color-BFS}(k,G,c,S,\tau)$, and $\textsf{color-BFS}(k,G[V\smallsetminus S,c,W,\tau)$ are each of at most $k\tau$ rounds. 
The for-loop at Instruction~\ref{inst:nb_colorings} is performed $K$ times. So, overall, the total number of rounds performed by Algorithm~\ref{algo-C2k-freeness} is $Kk\tau=O( \log^2(1/\eps) \cdot 2^{3k} k^{2k+3} n^{1-1/k})$ rounds.

\paragraph{Acceptance without error.}

First note that a node~$u$ may reject in Algorithm~\ref{algo-C2k-freeness} only while performing one of the calls to color-BFS. The rejection by~$u$ is thus caused by some identifier $\id(u_0)$  that has traversed two paths, on the one hand a path $u_0, u_1,\dots,u_k$, and on the other hand a path $u_0,u_{2k-1},u_{2k-2},\dots,u_k$, where $u_i$ is colored~$i$, for every $i\in\{0,\dots,2k-1\}$, and has reached~$u=u_k$ from both neighbors $u_{k-1}$ and $u_{k+1}$. These paths form together a $2k$-cycle $u_0,\dots,u_{k-1},u_k,u_{k+1},\dots, u_{2k-1}$. Therefore, any node that rejects does so rightfully. In other words, if $G$ contains no $2k$-cycles, then the probability that all nodes accept is~1, as desired.

\paragraph{Rejection probability.}

We now prove that when there is a $2k$-cycle $C$ in $G$, Algorithm~\ref{algo-C2k-freeness} rejects with probability $1-\eps$.
We do the analysis by considering three cases, not necessarily disjoint but covering every possible scenarios. Each of them will reject with probability at least $1-\eps$, leading to the claimed global rejection. The probability events will be analyzed using Facts~\ref{fact:good_colors}, \ref{fact:size_S}, \ref{fact:deg_S_u_0}.
For any considered $2k$-cycle $C=(u_0,\dots,u_{2k-1})$, 
we will assume that $c(u_i)=i$ for every $i\in\{0,\dots,2k-1\}$.
This is indeed the case, with probability at least $1-\nicefrac{\eps}{3}$,
for at least one coloring of the loop (Fact~\ref{fact:good_colors}
with $\alpha=\hat\eps$).
Moreover, we will also assume that $|S|\leq 4k^2\hat\eps n^{1-1/k}$, since this occurs with probability at least $1-\nicefrac{\eps}{3}$ (Fact~\ref{fact:size_S} with $\alpha=2k^2\hat\eps$, using that $k\geq 2$).
Lastly, when considering $u_0\in C$ with $\deg(u_0)\geq n^{\nicefrac{1}{k}}$, we will assume that $|N_G(u_0)\cap S| \geq k^2$, since this occurs with probability at least $1-\nicefrac{\eps}{3}$ (Fact~\ref{fact:deg_S_u_0} with $\alpha=2k^2 \hat\eps$).

\begin{description}
    
\item[Case 1:] $C\subseteq G[U]$. Let $C=(u_0,\dots,u_{2k-1})$ with $u_i\in U$ for every $i\in\{0,\dots,2k-1\}$. By Lemma~\ref{lemma:success_U}, when a coloring of the for-loop satisfies $c(u_i)=i$ for every $i\in\{0,\dots,2k-1\}$, node $u_k$ rejects in $\textsf{color-BFS}(k,G[U],c,U,\tau)$.
Thus Algorithm~\ref{algo-C2k-freeness} rejects with probability at least $1-\nicefrac{\eps}{3}$.
    
\item[Case 2:] $C\cap S\neq\varnothing$. Let  $C=(u_0,\dots,u_{2k-1})$ with $u_0\in S$. 
If a coloring of the for-loop satisfies $c(u_i)=i$ for every $i\in\{0,\dots,2k-1\}$,
and if $|S|\leq4k^2\hat\eps n^{1-1/k}$, then, by Lemma~\ref{lemma:success_S}, $u_k$ rejects in $\textsf{color-BFS}(k,G,c,S,\tau)$. 
Thus Algorithm~\ref{algo-C2k-freeness}  rejects with probability at least $1-\nicefrac{2\eps}{3}$, using union bound.

\item[Case 3:] $C\cap S=\varnothing$ and $C\cap (V\smallsetminus U)\neq\varnothing$. Let $C=(u_0,\dots,u_{2k-1})$ with  $u_0\in V\smallsetminus(S\cup U)$, in particular $\deg(u_0)\geq n^{\nicefrac{1}{k}}$. 
Using union bound, with probability at least $1-\eps$, 
we get (1)~a coloring of the for-loop satisfying $c(u_i)=i$ for every $i\in\{0,\dots,2k-1\}$,
(2)~$|S|\leq 4k^2\hat\eps n^{1-1/k}$,
and (3)~$|N_G(u_0)\cap S|\geq k^2$.
Thanks to Lemma~\ref{lemma:success_W}, since $u_0\in W_0$, 
$u_k$ rejects in $\textsf{color-BFS}(k,G[V\smallsetminus S],c,W,\tau)$.
\end{description}
The analysis of these three cases completes the proof. 
\end{proof}
\end{toappendix}

\section{Quantum complexity}\label{sec5}

The objective of this section is to prove that by using an
adaptation of Grover search to the distributed setting,
Algorithm~\ref{algo-C2k-freeness} can be sped up to $\tilde O(n^{\nicefrac{1}{2}-\nicefrac{1}{2k}})$ rounds in a quantum setting, as claimed in Theorem~\ref{quantum_theo_simplified}. The whole section is devoted to the proof of Theorem~\ref{quantum_theo_simplified}\withappendix{, along with the lower bounds and the case of odd cycles in Sections~\ref{subsec:lowerboundproof} and~\ref{subsec:upperboundodd} in Appendix}. 

%

\subsection{Preliminaries}

\subsubsection{Quantum amplification}

\withoutappendix{We first review the framework for distributed quantum search~\cite{GallM18} in the simple case of a search problem (instead of a general optimization problem),
which can be seen as a distributed implementation of Grover's search~\cite{Grover96}, and its generalizations such as amplitude amplification~\cite{BrassardHT98}.
Then we show how this can be used 
as in the sequential setting (see \cite[Theorem~2.2]{Amb04}) for boosting the success probability of any one-sided error distributed algorithm.}\withappendix{We first adapt the framework of distributed quantum search~\cite{GallM18} for boosting the success probability of any one-sided error distributed algorithm, as in the sequential setting (see~\cite[Theorem~2.2]{Amb04}).}
Those procedures require a centralized control, thus a specific node $\leader$ will have to play that role in distributed environments. 

\begin{toappendix}
\withappendix{\subsection{From Grover's search to Monte-Carlo amplification}\label{app-mca}}
\withappendix{We first review the setting of Grover's search.}
Let $X$ be a finite set, and let $f\colon X\to \{0,1\}$. 
Let $\leader$ be any fixed vertex of a network (e.g., an elected leader), whose purpose is to find $x\in X$ such that $f(x)=1$ (assuming such elements exist). Vertex $\leader$ has  two (randomized or quantum) distributed procedures at its disposal:
\begin{itemize}
\item $\setup$: Sample $x\in X$ (or create a superposition of elements $x\in X$) such that $f(x)=1$ with probability $p_\find$ (when measuring $x$ in the quantum setting). We let $T_{\setup}$ denote the round-complexity of $\setup$. 
\item $\checking$: Compute $f(x)$, given input $x\in X$. We let $T_{\checking}$ denote the round-complexity of $\checking$.
\end{itemize}

Assume one wants to amplify the success probability of $\setup$ as follows:
(1) vertex $\leader$ has to sample $x\in X$ in the support of $\setup$,
and (2) $x$ should satisfy $f(x)=1$ with high probability whenever $p_\find\geq \eps$.
A randomized strategy consists in executing $\setup$ $\Theta(1/\eps)$ times,
and returning any of the sampled values of $x$ satisfying $f(x)=1$, using $\checking$.
The success probability can be made arbitrarily high assuming that $p_\find\geq \eps$.
This procedure has round complexity  $O((T_{\setup}+T_{\checking})/\eps)$.
But quantumly we can do quadratically better in $\eps$.

\begin{lemma}[Distributed quantum search~\mbox{\cite[Theorem~7]{GallM18}}]\label{lem:distGrovGene}
Let $f,\leader$, $\setup,\checking$, $T_\setup, T_\checking$ and $p_\find$ be defined as above.
For any $\delta>0$, there is a quantum distributed algorithm 
with round complexity $$O\left(\log(1/\delta) \cdot \frac1{\sqrt{\eps}}\left(T_{\setup}+T_{\checking}\right)\right),$$ such that (1)~node $\leader$ returns $x$ in the support of $\setup$,
and (2)~$f(x)=1$ with probability at least $1-\delta$ whenever $p_\find\geq\eps$.
\end{lemma}

\paragraph{Remark.}

 Note that, actually,  the statement above slightly differs from Theorem~7 in~\cite{GallM18} in several aspects.
First, it is restricted to the simplest case of search problems (instead of optimization ones).
Second, there is no mention of any initialization procedure that we let outside our framework for the sake of simplifying the presentation. Third, procedures $\setup$ and $\checking$  may be  deterministic or randomized. Indeed, in the distributing setting, they can be converted into quantum procedures 
using standard techniques similar to those used for the sequential setting (see~\cite{GallM18}).
Last, point~(1) was not explicit in~\cite{GallM18}, but this is a well-known property of Grover's search and amplitude amplification.
\end{toappendix}

\withoutappendix{As an application, and similarly to the sequential case (see for instance \cite[Theorem~2.2]{Amb04}), Lemma~\ref{lem:distGrovGene} can be used to amplify the \emph{success} probability of any one-sided error Monte-Carlo distributed algorithm.}
This amplification result is of independent interest, and may be used for other purposes, beyond the design of $H$-freeness decision algorithms. We call it \emph{Distributed quantum Monte-Carlo amplification}
\withappendix{(see proof in Appendix~\ref{app-mca})}.

\begin{theorem}[Distributed quantum Monte-Carlo amplification]\label{theo:distampli}
Let $\mathcal{P}$ be a Boolean predicate on graphs.
Assume that there exists a (randomized or quantum) distributed algorithm $\mathcal{A}$ that decides $\mathcal{P}$ with one-sided \emph{success} probability $\eps$ on input graph $G$, i.e., 
\begin{itemize}
    \item If $G$ satisfies $\mathcal{P}$, then, with probability $1$, $\mathcal{A}$ accepts at all nodes;
    \item If $G$ does not satisfy $\mathcal{P}$, then, with probability at least $\eps$, $\mathcal{A}$ rejects in at least one node.
\end{itemize}
Assume further that $\mathcal{A}$  has round-complexity $T(n,D)$ 
for $n$-node graphs of diameter at most $D$.
Then, for any $\delta>0$, there exists a quantum  distributed algorithm  $\mathcal{B}$ 
that decides $\mathcal{P}$ with one-sided \emph{error} probability~$\delta$,
and round-complexity $\mathrm{polylog}(1/\delta)\cdot \frac{1}{\sqrt{\eps}}(D+T(n,D))$.
\end{theorem}

\begin{toappendix}
\withappendix{\medskip\\ \indent We can now prove Theorem~\ref{theo:distampli}.}
\begin{proof}
Let $\mathcal{A}$ be the given distributed algorithm.
First we recast the problem in the framework of Lemma~\ref{lem:distGrovGene}.
Let $X=\{\mathit{accept},\mathit{reject}\}$
and $f\colon X\to \{0,1\}$ be such that $f(\mathit{accept})=0$ and $f(\mathit{reject})=1$.
Now define $\setup$ as an algorithm which (1) selects a leader node $\leader$, (2) runs $\mathcal{A}$, (3) broadcasts the existence of a rejecting node to $\leader$, (4) outputs $\mathit{accept}$ when all nodes accepts, and
$\mathit{reject}$ otherwise. Thus $T_{\setup}=T(n,D)+O(D)$.
The procedure $\checking$ is trivial since $\leader$ simply transforms  $\mathit{accept}$ in $0$ and $\mathit{reject}$ in $1$, which requires no rounds, thus $T_{\checking}=0$.
Let us call $\mathcal{B}$ the resulting algorithm after applying Lemma~\ref{lem:distGrovGene} with these procedures, and with the same parameter~$\eps$.
By construction, $\mathcal{B}$ has the claimed round-complexity. For the analysis of correctness, we consider two cases. 
\begin{itemize}
    \item Assume that $\mathcal{A}$ accepts with probability~$1$, that is, $\mathcal{A}$ samples $\mathit{reject}$  with probability $0$. Then $\mathcal{B}$ also samples $\mathit{reject}$ with probability $0$ since it samples in the support of $\mathcal{A}$.

    \item Assume now that $\mathcal{A}$ rejects with probability at least $\eps$. Then $\mathcal{A}$ samples $\mathit{reject}$ with the same probability, that is $p_\find\geq \eps$.
Thus $\mathcal{B}$ will sample $\mathit{reject}$ with probability at least $1-\delta$.
\end{itemize}
We conclude that $\mathcal{B}$ is an algorithm that solves $\mathcal{P}$ with one-sided error probability $\delta$.
\end{proof}
\end{toappendix}

\subsubsection{Diameter reduction}

As demonstrated in \cite{eden2022sublinear}, one can assume that the network has small diameter when looking for a forbidden connected subgraph.
As observed in~\cite{Censor-HillelFG22,ApeldoornV22}, this assumption is valid for both quantum and randomized algorithms in the CONGEST model \withappendix{(see Appendix~\ref{app-dr} for more details)}.

\begin{lemma}[\mbox{\cite[Theorem~15]{eden2022sublinear}}]\label{theo:diam_red}
Let $H$ be any fixed, connected $k$-vertex graph. Let $\mathcal{A}$ be a randomized (resp., quantum) algorithm that decides $H$-freeness, with round-complexity $T(n,D)$ in $n$-node graphs of diameter at most~$D$, and error probability $\rho(n) = o(1/(n\log n))$. Then there is a randomized (resp., quantum) algorithm $\mathcal{A}'$ that decides $H$-freeness with round-complexity $\mathrm{polylog}(n)(T(n,O(k\log n)) + k)$, and error probability at most $(c\rho(n)\cdot n\log n+1/\mathrm{poly}(n))$, for some constant $c>0$.
\end{lemma}

\begin{toappendix}
\withappendix{\subsection{Diameter reduction (Lemma~\ref{theo:diam_red})\label{app-dr}}}
This result is based on a powerful technique of network decomposition from~\cite{ELKIN2022150}, which leads to the following (classical) preprocessing step, where nodes get one or more colors.

\begin{lemma}[\mbox{\cite[Theorem~17]{eden2022sublinear}}]\label{theo:pre_diam_red}
Let $G = (V , E)$ be an $n$-node graph and let $k \geq 2$ be an integer. There is a randomized algorithm  with round complexity $ k \polylog( n)$ and error probability at most $1/\poly(n)$ that constructs a set of clusters of diameter $O(k \log n)$ such that (1)~each node is in at least one cluster, (2)~the clusters are colored with $\gamma=O(\log n)$ colors, and (3)~clusters of the same color are at distance at least~$k$ from each other in~$G$.
\end{lemma}

We now briefly review the reduction in Lemma~\ref{theo:diam_red} to convince the reader that this reduction applies to quantum protocols too. The following can be viewed as a sketch of proof of Lemma~\ref{theo:diam_red}.

\paragraph{Construction in the proof of Lemma~\ref{theo:diam_red}.} 

We define $\mathcal{A'}$ as follows.
First $\mathcal{A'}$ computes the network decomposition of Lemma~\ref{theo:pre_diam_red} with parameter $2k+1$. Therefore, the clusters of a same color are at pairwise distance at least $2k+1$. Define $G(i,k)$ as the graph induced by all vertices of color $i\in [\gamma]$, and their $k$-neighborhood. Observe that each connected component of $G(i,k)$ is of diameter $O(k \log n)$. Indeed when we enlarge the clusters of color $i$ as in $G(i,k)$ by adding their $k$-neighborhood, the subgraphs of $G(i,k)$ resulting from different clusters are disjoint. Moreover, they are not connected in $G(i,k)$. Therefore each connected component of $G(i,k)$ is obtained by the enlargement of one cluster, with its $k$-neighborhood. Thus the diameter of the component is at most the diameter of the cluster, plus $2k$.

Then, sequentially, for each color $i\in [\gamma]$, $\mathcal{A'}$ runs $\mathcal{A}$ in parallel on each connected component of $G(i,k)$. For the analysis, we first assume that the network decomposition has succeeded, and satisfies the conclusions of Lemma~\ref{theo:pre_diam_red}.
%
%
Since connected components of $G(i,k)$ have diameter $O(k \log n)$, and since we used $O(\log n)$ different colors, the round-complexity of $\mathcal{A'}$ is as claimed. 
The correctness of $\mathcal{A'}$ is a direct consequence of the following observation. $G$ contains a copy of $H$ as a subgraph if and only if there exists a color $i\in[c]$ such that $G(i,k)$ contains a copy of $H$ as a subgraph.
For the overall error probability, simply take the union bound of all the potential failure events.
\end{toappendix}

\subsection{Quantum algorithm}

\subsubsection{Congestion reduction}

Our classical algorithm has a constant success probability (cf. Theorem~\ref{classical_theo}), and its three calls to the subroutine $\textsf{color-BFS}$ have complexities upper bounded by $O(n^{1-1/k})$. 
In order to speed up this algorithm, we first
reduce its success probability to $\eps=\Theta(1/n^{1-1/k})$
such that the $\textsf{color-BFS}$ subroutines have $O(1)$ round-complexities.
Classically, one would have to repeat this process ${1}/{\eps}$ times in order to get constant success probability,
whereas ${1}/{\sqrt{\eps}}$ quantum rounds suffices thanks to Theorem~\ref{theo:distampli}.

Recall that the complexity $O(n^{1-1/k})$ of $\textsf{color-BFS}$ is a consequence of the setting of the threshold $\tau= O(n^{1-1/k})$.
Indeed, every node only needs to send information that comes from at most $\tau$ nodes $x_0$, leading to a congestion of $\tau$, and thus a round complexity of $\tau$.
By activating each node $x_0$ with probability $\eps=O(\nicefrac{1}{\tau})$ uniformly at random, one can expect to reduce the congestion to $ O(\eps\tau)= O(1)$, and to reduce the success probability to~$\eps$.
Then, we apply Theorem~\ref{theo:distampli} to get constant success probability
within $O(D\times \sqrt{\tau})$ rounds, where $D$ is the graph diameter.
Note that, later on, we will eliminate the diameter dependence by employing the standard diameter reduction offered by Lemma~\ref{theo:diam_red}.
This technique has already been used for detecting $C_4$ with round complexity $\tilde O(n^{1/4})$ in an unpublished work~\cite{censor2023discussion}, using directly Lemma~\ref{lem:distGrovGene},
and not our new Theorem~\ref{theo:distampli}. The implementation is then a bit more involved.
Each node initially samples a private token in $[{1}/{\eps}]$, activated nodes are then decided by the leader choosing a token and broadcasting it. While the token sampling is decentralized, the activation of nodes of a given token must be centralized before applying Lemma~\ref{lem:distGrovGene}, which yields also a multiplicative term~$D$.
However, for our purpose, we used a novel, and more general randomized distributed algorithm (Theorem~\ref{classical_theo}), together with a more encapsulated quantum framework (Theorem~\ref{theo:distampli}) that eases its application.

\subsubsection{Application to \textsf{color-BFS}}

To speedup Algorithm~\ref{algo-C2k-freeness}, 
we replace $\textsf{color-BFS}(k,H,c,X,\tau)$
with a randomized protocol called $\textsf{randomized-color-BFS}(k,H,c,X,\tau)$, described in Algorithm~\ref{algo_random_color} \withappendix{in Appendix~\ref{app-rbvs}}.
This randomized protocol  has smaller round complexity, and thus has smaller success probability.
Compared to $\textsf{color-BFS}(k,H,c,X,\tau)$, the modifications are the following: 
\begin{itemize}
    \item A node colored~0 does not systematically initiate a search, but does so with probability~$1/\tau$ (cf. Instruction~\ref{rdm_ins:BFS-initialization}). 
    \item Instead of using the threshold $\tau$, the randomized algorithm uses a much smaller (constant) threshold, merely equal to~4 (cf. Instruction~\ref{rdm_ins:I_size}). 
\end{itemize}

\begin{toappendix}
\withappendix{\subsection{Randomized version of \textsf{color-BFS}\label{app-rbvs}}}
\begin{algorithm}[ht] 
\caption{$\textsf{randomized-color-BFS}(k,H,c,X,\tau)$}
\label{algo_random_color}
\begin{algorithmic}[1]
\State every node $x\in X$ with $c(x)=0$ 
sends its ID to its neighbors in~$H$ with probability $1/\tau$\label{rdm_ins:BFS-initialization}
\For{$i=1$ to $k-1$}
\For{every node $v \in V(H)$ with $c(v)=i$ (resp., $c(v)=2k-i$) }
    \State $I_v\gets \{\mbox{IDs received from neighbors in~$H$ colored $i-1$ (resp. $2k-i-1)$}\}$\label{rdm_ins:Iv}
    \If{$|I_v|\leq 4$}\label{rdm_ins:I_size}
        \State $v$ forwards $I_v$ to its neighbors in~$H$ colored~$i+1$ (resp. $2k-i-1)$
    \EndIf
\EndFor
\EndFor
\For{every node $v\in V(H)$ colored $k$} 
    \If {$v$ receives a same ID from two neighbors respectively colored~$k-1$ and~$k+1$}\label{rdm_ins:reject}
        \State $v$ outputs \emph{reject}
    \EndIf
\EndFor
\end{algorithmic}
\end{algorithm}
\end{toappendix}


The protocol $\textsf{randomized-color-BFS}$ has round-complexity $O(1)$ since no node ever forwards more than $4$ identifiers. Let us analyse its success probability.
In order to take into account the fact that the first instruction of $\textsf{randomized-color-BFS}(k,H,c,X,\tau)$ is randomized,
we define the random set
$$X'_0=\{x\in X_0 \mid x\text{ sends $\id(x)$  at Instruction~\ref{rdm_ins:BFS-initialization} of Algorithm~\ref{algo_random_color}}\}.$$
For $X'_0(v)$, we proceed accordingly to the definition of $X_0(v)$ in Eq.~\eqref{eq:X_0_small_color} and~\eqref{eq:X_0_big_color}
for every $v\in H$ colored $i$, or $2k-i$, for every $i\in\{1,\dots,k-1\}$.
More formally, 
$$X'_0(v)=\{x\in X_0(v) \cap X_0'\}.$$
Then $X'_0(v)$ satisfies a fact similar to Fact~\ref{fact:link_I_X} for $X_0(v)$ \withappendix{(see Fact~\ref{rdm_fact:link_I_X} in Appendix~\ref{app-rbvs})}.

\begin{toappendix}
\begin{fact} \label{rdm_fact:link_I_X} ~
\begin{itemize}
    \item Let $(v_1,v_2,\ldots,v_{k-1})$ be a path in $H$ such that $c(v_i)=i$ for evey $i\in \{1,\dots,k-1\}$.
    If $|X'_0(v_{k-1})|\leq 4$
    then, for every $i\in\{1,\dots,k-1\}$,  $|I_{v_i}|\leq 4$ and $I_{v_i} = X'_0(v_i)$. 
    \item Let $(v_{2k-1},v_{2k-2},\ldots,v_{k+1})$ be a path in $H$ such that $c(v_i)=i$ for every  $i\in\{k+1,\dots,2k-1\}$.
If $|X'_0(v_{k+1})|\leq 4$
then, for every $i\in\{k+1,\dots,2k-1\}$, $|I_{v_i}|\leq 4$ and $I_{v_i} = X'_0(v_i)$.
\end{itemize}
\end{fact}

\begin{proof}
   We prove the result only for nodes colored between 1 and $k-1$ as, by symmetry, the same holds for nodes colored between $2k-1$ and $k+1$.
    It follows from Instructions~\ref{rdm_ins:BFS-initialization} and~\ref{rdm_ins:Iv} that, for every  $v\in H$, each identifier in $I_v$ corresponds to a node in $X'_0(v)$. Now, along a path $(v_1,\dots,v_{k-1})$, as long as $|X'_0(v_i)|\leq 4$ for every $i\in\{1,\dots,k-1\}$, the condition of Instruction~\ref{rdm_ins:I_size} is always satisfied for $I_{v_i}$, and every node in $X'_0(v_i)$ belongs to $I_{v_i}$. The result follows thanks to the fact that $X'_0(v_i)\subseteq X'_0(v_{k-1})$ for every $i\in\{1,\dots,k-1\}$.
\end{proof}
\end{toappendix}

We can now state a first statement on the success probability of \textsf{randomized-color-BFS}, 
which is not yet suitable for applying Theorem~\ref{theo:distampli}. Note that the two cases considered in the lemma below do not cover all scenarios. 

\begin{lemma}\label{lem:rdm_color-BFS}
Protocol $\textsf{randomized-color-BFS}(k,H,c,X,\tau)$ satisfies the following. 
\begin{itemize}
\item If there are no $2k$-cycles in $H$, then,  with probability~$1$, all nodes accept.

\item If there is a $2k$-cycle $C=(u_0,\dots,u_{2k-1})$ in $H$, with $u_0\in X$, $c(u_i)=i$ for every ${i\in\{0,\dots,2k-1\}}$, 
$|X_0(u_{k-1})|\leq\tau$, and $|X_0(u_{k+1})|\leq\tau$,
then, with probability at least $1/(2\tau)$, at least one node rejects. 
   \end{itemize}
\end{lemma}

\begin{proof}
The first claim is straightforward. 
Let us prove the second claim. As each element of $X_0$ belongs to $X_0'$ with probability $\nicefrac{1}{\tau}$, we get $u_0\in X_0'$ with probability $\nicefrac{1}{\tau}$. Assume that $u_0\in X_0'$.  Then  $u_0\in I_{u_1}\cap I_{u_{2k-1}}$. Conditioned to $u_0\in X_0'$, the expected sizes of $X'_0(u_{k-1})\smallsetminus\{u_0\}$ and $X'_0(u_{k+1})\smallsetminus\{u_0\}$ are at most~$1$. Thus, by Markov's inequality, each of them is at least~$4$ with probability at most~$1/4$. Therefore by the union bound, 
$$\Pr\Big[\big(|X'_0(u_{k-1})\smallsetminus\{u_0\}|\leq 3\big)\land \big(|X'_0(u_{k+1})\smallsetminus\{u_0\}|\leq 3\big)\Big]\geq\frac{1}{2} 
.$$     
It follows that, with probability at least $\nicefrac{1}{2\tau}$, we have $u_0\in I_{u_1}\cap I_{u_{2k-1}}$, $|X'_0(u_{k-1})|\leq 4$, and $|X'_0(u_{k+1})|\leq 4$.
Using Fact~\ref{rdm_fact:link_I_X}, we derive  that, for every $i\in\{1,\dots,k-1\}$, $|I_{u_i}|\leq 4$ and $|I_{u_{2k-i}}|\leq 4$, meaning that $u_k$ rejects after receiving $u_0$'s identifier along path $(u_1,\dots,u_{k-1})$, and along path $(u_{2k-1},\dots,u_{k+1})$. 
\end{proof}


\subsubsection{Application to cycle detection} 

We now have all the ingredients to prove the upper bound in Theorem~\ref{quantum_theo_simplified} regarding deciding $C_{2k}$-freeness in quantum \CONGEST. 


\begin{lemma}\label{lem:quantum-c2k}
For every $k\geq 2$, there is a randomized distributed algorithm $\mathcal{A}$ solving $C_{2k}$-freeness with one-sided success probability $1/(3\tau)$, and running in $k^{O(k)}$ rounds. 
\end{lemma}

\begin{proof}
    Let $\mathcal{A}$ be Algorithm~\ref{algo-C2k-freeness} 
    where $\textsf{color-BFS}$ is replaced by $\textsf{randomized-color-BFS}$.
    To analyze $\mathcal{A}$ we refer to the analysis of Algorithm~\ref{algo-C2k-freeness} with $\epsilon=\nicefrac{1}{3}$ in Theorem~\ref{classical_theo} \withappendix{(cf. Appendix~\ref{app-thm1})}, and only highlight the differences. 
    \begin{itemize}
        \item Complexity: By construction, the new threshold $4$ in \textsf{randomized-color-BFS}, instead of $\tau$ in \textsf{color-BFS}, ensures that the overall complexity is $4kK$, that is $O(k(2k)^{2k})=k^{O(k)}$ rounds.
  
        \item Acceptation probability: As a node can only reject in $\textsf{randomized-color-BFS}(k,H,c,X,\tau)$, then by Lemma~\ref{lem:rdm_color-BFS}, Algorithm~$\mathcal{A}$ always accepts when there are no $2k$-cycle in $G$.

        \item Rejection probability: We prove that if there is $2k$-cycle $C$ in $G$ then Algorithm~$\mathcal{A}$ rejects with probability $\Omega(\nicefrac{1}{\tau})$.  Observe that the proofs of Lemmas~\ref{lemma:success_U},~\ref{lemma:success_S} and~\ref{lemma:success_W} consist in upper bounding by $\tau$ the size of the sets $X_0(u_{k-1})$ and $X_0(u_{k+1})$, where $C=(u_0,\dots,u_{2k-1})$ is the considered cycle. In particular, whenever those lemmas are applied, we can use Lemma~\ref{lem:rdm_color-BFS} instead, on the same considered cycle, in order to lower bound the rejection probability by~$1/(2\tau)$. With this modification in mind, the rest of the proof of Theorem~\ref{classical_theo} applies, leading to a rejection probability of at least $(1-\eps)/(2\tau)=1/(3\tau)$.
    \end{itemize}
This completes the proof. 
\end{proof}

The upper bound for even cycles in Theorem~\ref{quantum_theo_simplified} is now proved in the following lemma. 

\begin{lemma}\label{lem:quantum_res}
There is a quantum distributed algorithm $\mathcal{A}$ solving $C_{2k}$-freeness with one sided error probability $\nicefrac{1}{\text{poly}(n)}$, and running in $k^{O(k)}\cdot \mathrm{polylog}(n) \cdot n^{\nicefrac{1}{2}-\nicefrac{1}{2k}} $ rounds. 
\end{lemma}

\begin{proof}
First, by Lemma~\ref{lem:quantum-c2k}, we get a randomized algorithm with one-sided success probability~$\nicefrac{1}{3\tau}$, and round-complexity~$k^{O(k)}$.
Then, we apply Theorem~\ref{theo:distampli} to amplify this algorithm to get a one-sided error probability $\nicefrac{1}{\text{poly}(n)}$, in $\mathrm{polylog}(n)\cdot \sqrt\tau\cdot(k^{O(k)}+D)$
rounds.
Finally, we get rid of the diameter factor by applying Lemma~\ref{theo:diam_red}, which yields one-sided error probability $\nicefrac{1}{\text{poly}(n)}$, and round-complexity 
\[
\mathrm{polylog}(n)[(k^{O(k)}+O(k\log n))\sqrt{\tau}+2k]= k^{O(k)}\mathrm{polylog}(n)\cdot\sqrt{\tau}=k^{O(k)}\mathrm{polylog}(n) \cdot n^{\frac{1}{2}-\frac{1}{2k}},
\]
which completes the proof.
\end{proof} 

\withappendix{For the upper bound for odd cycles and the lower bounds in Theorem~\ref{quantum_theo_simplified}, we refer to the Appendix (cf. Sections~\ref{sec:qlb}, \ref{subsec:upperboundodd}, and \ref{subsec:collec_of_cycles}).}

\begin{toappendix}
\subsection{Quantum Lower bounds}\label{sec:qlb}
\label{subsec:lowerboundproof}

In this section, we prove the lower bounds stated in Theorem~\ref{quantum_theo_simplified}. Let us first consider even length cycles. 

\subsubsection{Even-Length Cycles}

We show that, for any $k\geq 2$, deciding $C_{2k}$-freeness requires $\Omega(n^{1/4}/\log n)$ rounds in the quantum $\CONGEST$ model. 
The proof relies on a reduction from the Set-Disjointness problem  in the two-party quantum communication framework,  to the $C_{2k}$-freeness problem in the quantum $\CONGEST$ model. The reduction is the same as in the classical case.
For $k=2$, it is based on the construction of~\cite{DruckerKO13},
and for $k\geq 3$, on  the one of~\cite{KorhonenR17}.
Recall that, in the Set-Disjointness problem, each of two players, typically referred to as Alice and Bob,  gets a subset of a universe of size $N$, say $[N]$, and they must decide whether there is an element in common in their subsets, while exchanging as few qubits (in the quantum setting) as possible. The total number of exchanged qubits is called the quantum communication complexity. The number of rounds of interactions between the two players is the round complexity. Note that the input of Alice (resp., Bob) can be represented as a binary string $x\in\{0,1\}^N$ (resp., $y\in\{0,1\}^N$), where $x_i=1$ (resp., $y_i=1$) whenever element~$i$ is in the input set of Alice (resp., Bob).

The reduction in \cite{DruckerKO13} and \cite{KorhonenR17} are based on picking a gadget graph~$G$ with $N$ edges, $e_1,e_2,\dots,e_N$, and constructing a graph $H$ composed on two subgraphs $G_A$ and $G_B$ of~$G$ connected by a perfect matching. These subgraphs are obtained from the inputs $x$ and $y$ of Alice and Bob, as follows. For every $i\in [N]$, Alice (resp., Bob) keeps edge $e_i$ if $x_i=1$ (resp., $y_i=1$), and discards it otherwise. The gadget graphs in \cite{DruckerKO13} and \cite{KorhonenR17} are different, but the construction is the same once the gadget graph is fixed. The gadget graph in \cite{DruckerKO13} has $N=\Theta(n^{3/2})$ edges, while the gadget graph in \cite{KorhonenR17} has $N=\Theta(n)$ edges. Let us denote by $H'$ the graph resulting from the former, and by $H''$ the graph resulting from the latter. 
\begin{itemize}
    \item It was shown in \cite{DruckerKO13} that if there is a $T(n)$-round \CONGEST\/ algorithm deciding $C_4$-freeness in the $n$-node graphs $H'$ then there is a $T(n)$-round two-party communication protocol for Set-Disjointness of size $N=\Theta(n^{3/2})$ with communication complexity $O(T(n) \cdot n\cdot \log n)$ bits.
    
    \item Similarly, it was shown in \cite{KorhonenR17} that  if there is a $T(n)$-round \CONGEST\/ algorithm deciding  $C_{2k}$-freeness for $k\geq 3$ in the $n$-node graphs $H''$ then there is a $T(n)$-round two-party communication protocol for Set-Disjointness of size $N=\Theta(n)$ with communication complexity $O(T(n)\cdot \sqrt{n}\cdot \log n)$ bits.
\end{itemize}
By exactly the same arguments, the same two properties hold for quantum \CONGEST\/ algorithms. Now, it is known~\cite{braverman18} that, in the quantum two-party communication model, for every $r\geq 1$, any $r$-round communication protocol solving Set Disjointness for sets of size $N$ has communication complexity $\Omega(r+\frac{N}{r})$ qubits. As a consequence, we get the following. 
\begin{itemize}
    \item For $k=2$, any quantum algorithm solving $C_4$-freeness in $T(n)$ rounds must satisfy $$T(n)\cdot n \cdot \log n = \Omega(N/T(n)),$$ with $N=\Theta(n^{3/2})$, and therefore $T(n) = \Omega(n^{{1}/{4}}/\sqrt{\log n})$.

    \item For $k\geq 3$, any quantum algorithm solving $C_{2k}$-freeness in $T(n)$ rounds must satisfy $$T(n)\cdot\sqrt{n}\cdot\log n = \Omega (N/T(n)),$$ with $N=\Theta(n)$. Therefore $T(n)= {\Omega}(n^{1/4}/\sqrt{\log n})$, as claimed.
\end{itemize}

\subsubsection{Odd-Length Cycles}

Let us now move on with establishing the lower bound stated in Theorem~\ref{quantum_theo_simplified} for cycles of odd lengths. That is, we show that, for any $k\geq 2$, the round complexity of $C_{2k+1}$-freeness is $\tilde\Theta(\sqrt{n})$ in quantum $\CONGEST$.  Once again, we use reduction from set disjointness. A gadget graph $H''$ with $\Theta(n^2)$ edges has been constructed in \cite{DruckerKO13}, for which it was proved that if there is a $T(n)$-round \CONGEST\/ algorithm deciding $C_{2k+1}$-freeness in the $n$-node graphs $H''$ then there is a $T(n)$-round two-party communication protocol for Set-Disjointness of size $N=\Theta(n^2)$ with communication complexity $O(T(n) \cdot n\cdot \log n)$ bits. The same holds for quantum \CONGEST. Therefore, using again the lower bound for set disjointness in~\cite{braverman18}, we get that any quantum algorithm solving $C_{2k+1}$-freeness in $T(n)$ rounds must satisfy $$T(n)\cdot n\cdot \log n = \Omega (N/T(n)),$$ with $N=\Omega(n^2)$. Therefore $T(n) = \Omega(\sqrt{n}/\log n)$, as claimed.

\subsection{Quantum Upper Bound for Deciding Odd Cycles}
\label{subsec:upperboundodd}

We finally show that the lower bound for odd cycles $C_{2k+1}$ with $k\geq 2$ established in the previous section is tight, which completes the proof of Theorem~\ref{quantum_theo_simplified}. Specifically, we show that there is a simple one-sided error randomized algorithm with success probability $\Omega(1/n)$ for deciding $C_{2k+1}$-freeness. As for the even case, the algorithm consists of looking for a well colored cycle, using a repetition of $K=O(1)$ random coloring~$c$ of the vertices of the graph, but with colors taken in $\{0,\dots,2k\}$ (instead of $\{0,\dots,2k-1\}$). For every coloring~$c$, we apply a procedure similar to $\textsf{randomized-color-BFS}(k,G,c,V,4)$ described in Algorithm~\ref{algo_random_color}. Indeed, for odd cycles the procedure only differs by the fact that we look for a cycle $(u_0,\dots,u_{2k})$ instead of $(u_0,\dots,u_{2k-1})$. That is, a node colored~$k$ checks reception of a same identifier transmitted along a path colored $0,1,\dots,k-1,k$ (of length~$k$), and a path colored $0,2k,\dots,k+1,k$ (of length $k+1)$. For any node~$u$,  and for any coloring~$c$, we have $|V_0(u)|\leq |V|\leq n$. Therefore, by using the same arguments as in Lemma~\ref{lem:rdm_color-BFS}, we get an algorithm with one-sided success probability $\Omega(\nicefrac{1}{n})$, and constant round-complexity. Indeed, whenever a $(2k+1)$-cycle~$C$ is well colored, the node $u_0$ with color~0 in $C$ has probability $\nicefrac{1}{n}$ of sending its identifier. Moreover, if this happens then, with constant probability, no node of the cycle~$C$ will receive more than a constant number of identifiers. One can apply our quantum boosting technique (cf. Theorem~\ref{theo:distampli}), combined with the diameter reduction technique  (cf. Lemma~\ref{theo:diam_red}). This results in an algorithm with round-complexity $\tilde O(\sqrt{n})$, and error probability $1-\nicefrac{1}{\text{poly}(n)}$.

\subsection{Quantum Upper bound for Deciding Cycles of Bounded Length}
\label{subsec:collec_of_cycles}

For every $k\geq 2$, let $F_{2k}=\{C_\ell \mid 3\leq \ell \leq 2k\}$. In this section, we show how to use our quantum algorithm deciding $C_{2k}$-freeness for deciding $F_{2k}$-freeness, in $\tilde O(n^{\nicefrac{1}{2}-\nicefrac{1}{2k}})$ rounds. In a nutshell, the (quantum) algorithm in~\cite{ApeldoornV22} uses a different bound $d_{max}$ compared to \cite{Censor-HillelFG20}, and quantize only the search for heavy cycles. Instead, we keep the same bound $d_{max}=n^{1/k}$, but we quantize the search of both light and heavy cycles.  Our quantum algorithm is rejecting with probability $1-\nicefrac{1}{\mathrm{poly}(n)}$ if there is a cycle of length $\ell\in\{3,\dots,2k\}$, and is accepting otherwise. 

Our quantum algorithm results from quantizing the classical algorithm for $F_{2k}$-freeness in~\cite{Censor-HillelFG20}, in the same way we quantized our classical algorithm for $C_{2k}$-freeness (see Algorithm~\ref{algo-C2k-freeness}). More specifically, we sequentially check the existence of cycles for pairs of lengths,  by deciding $C_{2\ell-1}$- and $C_{2\ell}$-freeness conjointly, for every $\ell\in \{2,\dots,k\}$. For every~$\ell$, the decision algorithm works under the assumption that there are no cycles of length at most $2(\ell-1)$, as if there were such cycles, they would have been detected when testing a smaller pair of length. For a fixed~$\ell$, the algorithm is quasi-identical to Algorithm~\ref{algo-C2k-freeness} with just four differences, listed below (we refer to the instructions in Algorithm~\ref{algo-C2k-freeness}). 
\begin{itemize}
    \item Instruction~\ref{ins:Set-W}: We set $W$ as the set of \emph{all} neighbors of the set $S$, with no restrictions on the degrees of the nodes.
    \item Instruction~\ref{ins:set-threshold}: the threshold is now set to $\tau=2np$.
    \item Instructions~\ref{inst:colorbfs-S} and~\ref{inst:colorbfs-W} are merged into a single $\textsf{color-BFS}(m,G,c,W,\tau)$.
\end{itemize}
In addition, in any $\textsf{color-BFS}$ aiming at detecting $(2\ell-1)$-cycles, nodes colored $\ell+1$ also forwards the received identifiers to neighbors colored $\ell-1$, which reject if one of those identifiers is equal to one received from neighbors colored $\ell-2$.

The detection of light cycles by $\textsf{color-BFS}(m,G[U],c,U,\tau)$ performs in $O(n^{1-\nicefrac{1}{\ell}})$ rounds. On the other hand, regarding the simplified detection of heavy cycles using $\textsf{color-BFS}(m,G,c,W,\tau)$, if a node $v\in V$ receives more than $|S|$ identifiers of nodes in $W$, then two of them must be neighbors of a same node $s\in S$. The two sequentially colored paths from $s$ to $v$ induced by the forwarding of those two identifiers then form a cycle of length $\leq 2\ell$.
\end{toappendix}

\section{Conclusion}\label{sec:conclusion}

Thanks to our work, which complements previous work on the matter, the complexity landscape of deciding $C_k$-freeness in \CONGEST\/ is roughly as follows. In the classical setting, for $k>3$ odd, deciding $C_k$-freeness takes $\tilde\Theta(n)$ rounds, and, for $k\geq 4$ even, deciding $C_k$-freeness takes $\tilde O(n^{1-2/k})$ rounds. In the quantum setting, for $k>3$ odd, deciding $C_k$-freeness takes $\tilde\Theta(\sqrt{n})$ rounds, and, for $k\geq 4$ even, deciding $C_k$-freeness takes $\tilde O(n^{\nicefrac{1}{2}-\nicefrac{1}{k}})$ rounds. That is, quantum effects allowed us to design quantum algorithms with quadratic speedup compared to the best classical algorithms. Interestingly, in order to get quantum speed-up for even cycles, we first establish a trade-off between congestion, and round complexity. Such trade-offs might be exhibited for other distributed tasks, which would automatically lead  to a quantum speed-up for these tasks as well.

We have already mentioned the difficulty of designing lower bounds for triangle-freeness, as well as for $C_{2k}$-freeness for $k\geq 3$. Yet, it seems that for $C_{2k}$-freeness and $k\geq 3$, the design of classical algorithms with complexity $O(n^{1-1/k-\alpha})$, or quantum algorithms with complexity $O(n^{\nicefrac{1}{2}-\nicefrac{1}{2k}-\alpha})$, with $\alpha>0$, would require entirely new techniques.

Finally, it is worth mentioning that, as far as randomized algorithms for cycle detection are concerned, the randomized color-coding phases can often be replaced by deterministic protocols based on~\cite{ErdosHM64} (see, e.g., \cite{FraigniaudO19,KorhonenR17} for its application). However, as for most distributed algorithms detecting subgraphs, our algorithm needs to pick a set~$S$ of vertices at random. For 4-cycles, randomization is not necessary, but we do not known whether randomization is necessary or not for detecting larger cycles of even length in a sublinear number of rounds.


\bibliography{biblio-cycle}

\newpage 
\pagenumbering{roman}
\setcounter{page}{1}
\end{document}